\documentclass[11pt]{article}

\usepackage{amsmath}
\usepackage{amssymb}

\usepackage[a4paper,margin=1in]{geometry}
\usepackage{natbib}
\usepackage{amsthm}
\newtheorem{theorem}{Theorem}
\newtheorem{lemma}{Lemma}
\newtheorem{proposition}{Proposition}
\newtheorem{condition}{Condition}
\theoremstyle{remark}
\newtheorem{remark}{Remark}

\newcommand{\tbl}[2]{\caption{#1}\vspace{4pt}\par#2}
\newenvironment{tabnote}{\par\vspace{4pt}\footnotesize}{\par}
\newenvironment{keywords}{\par\noindent\textbf{Keywords: }}{\par}

\newcommand{\appendixone}{}
\newcommand{\appendixtwo}{}
\newcommand{\appendixthree}{}

\usepackage{placeins}
\usepackage{tikz}
\usetikzlibrary{arrows.meta,positioning,shapes.geometric,fit,backgrounds,calc,patterns}

\usepackage{hyperref}
\providecommand{\orcidlink}[1]{\href{https://orcid.org/#1}{\textsuperscript{\scriptsize iD}}}

\usepackage{listings}
\graphicspath{{./art/}}

\usepackage[plain,noend]{algorithm2e}

\makeatletter
\renewcommand{\algocf@captiontext}[2]{#1\algocf@typo. \AlCapFnt{}#2}

\def\@algocf@capt@plain{top}
\renewcommand{\algocf@makecaption}[2]{%
  \addtolength{\hsize}{\algomargin}%
  \sbox\@tempboxa{\algocf@captiontext{#1}{#2}}%
  \ifdim\wd\@tempboxa >\hsize%
    \hskip .5\algomargin%
    \parbox[t]{\hsize}{\algocf@captiontext{#1}{#2}}%
  \else%
    \global\@minipagefalse%
    \hbox to\hsize{\box\@tempboxa}%
  \fi%
  \addtolength{\hsize}{-\algomargin}%
}
\makeatother

\def\T{{ \mathrm{\scriptscriptstyle T} }}

\def\pr{\mathrm{pr}}
\usepackage{xspace}
\def\GPC{\textsc{gpc}\xspace}
\def\PFS{\textsc{pfs}\xspace}
\def\ORR{\textsc{orr}\xspace}
\def\ERADE{\textsc{erade}\xspace}
\def\DBCD{\textsc{dbcd}\xspace}
\def\RAR{\textsc{rar}\xspace}
\DeclareMathOperator{\var}{var}
\DeclareMathOperator{\cov}{cov}
\DeclareMathOperator{\avar}{avar}

\begin{document}

\title{Efficient response-adaptive randomization for multi-arm trials with prioritized composite endpoints}

\author{Ayon Mukherjee \\ \textit{Population Health Sciences Institute, Newcastle University, Newcastle upon Tyne, U.K.} \\ \texttt{ayon.mukherjee@newcastle.ac.uk}}
\date{}

\maketitle

\begin{abstract}
The efficient randomized-adaptive design attains the minimal possible variance of the allocation proportion for any pre-specified target and has been extended from two to several treatment arms, but both versions target a single response. Confirmatory trials increasingly rely on several prioritized endpoints, typically efficacy followed by safety, that cannot be reduced to one summary without losing information. The efficient randomized-adaptive design is extended here to targets driven by the net treatment benefit of generalized pairwise comparisons, allowing any number of arms and endpoints of mixed type, including censored time-to-event outcomes. Strong consistency, a law of the iterated logarithm, and asymptotic normality of the allocation proportions are established by combining a sequential Hájek projection for the net benefit with the martingale technique used for the original design, and the design is shown to attain the semiparametric efficiency bound within the class of designs that track a smooth function of the net benefit. In simulations calibrated to a three-arm phase III melanoma trial, the proposed design concentrates allocation on the superior treatment more efficiently than designs driven by a single endpoint, while preserving type~\textsc{i} error and power. A redesign of the trial illustrates the practical benefit, and the approach is discussed in relation to current regulatory thinking on adaptive designs for confirmatory trials.
\end{abstract}


\section{Introduction}
\label{sec:intro}

Response-adaptive randomization skews the allocation of patients toward superior treatments while the trial is in progress, so that more patients receive a better treatment without sacrificing the validity of subsequent inference \citep{RosenbergerLachin2016}. Among the many response-adaptive procedures proposed since \citet{Thompson1933} and \citet{ZelenPTW1969}, the efficient randomized-adaptive design of \citet{HuZhangHe2009} occupies a distinguished position: it is the first fully randomized procedure shown to attain the asymptotic lower bound on the variability of the allocation proportion, for essentially any pre-specified target allocation, including targets that depend on unknown response parameters. \citet{AlkhnefrHuZhai2025} extended the efficient randomized-adaptive design from two to $K\ (K\geq2)$ arms, answering an open problem stated by \citet{HuZhangHe2009}, and \citet{ZhangLX2026} has, contemporaneously with the present work, proposed a further general framework for efficient multi-arm designs attaining the same class of lower bounds.

Both of these generalizations, and indeed almost all of the response-adaptive randomization literature reviewed by \citet{RosenbergerLachin2016}, target an allocation proportion that is a smooth function of a single, scalar or low-dimensional, per-arm response. In many confirmatory trials, however, no single endpoint captures the clinical question. Regulators routinely require that a new treatment demonstrate favourable efficacy without unacceptable toxicity, and oncology trials in particular increasingly report several time-ordered or prioritized outcomes, such as progression-free survival, objective response, and high-grade toxicity, none of which the others can safely proxy \citep{ICH_E20_2022,FDA_adaptive_2019}. Confirmatory trials in practice frequently formalize this as co-primary endpoints, or as a primary endpoint accompanied by one or more key secondary endpoints on which the treatment claim also depends, so that a design responsive only to the leading endpoint discards exactly the information a regulator or clinician most wants reflected in how patients are allocated; this is the practical need the present development addresses. Generalized pairwise comparisons \citep{Buyse2010} were developed exactly for this setting: two patients, one from each of two arms, are compared endpoint by endpoint in a pre-specified order of clinical priority, ties at one endpoint being broken by the next, and the resulting net treatment benefit summarizes the probability that a random patient on the experimental arm fares better than a random patient on the control arm. The net benefit, and the closely related win ratio \citep{PocockEtAl2012}, have since been used to redesign and reanalyse a range of real trials with prioritized efficacy and safety outcomes \citep{PeronEtAl2016,DeBackerEtAl2024}, and its asymptotic distribution is now well understood as a two-sample $U$-statistic \citep{OzenneEtAl2021}, including under right censoring \citep{PeronEtAl2018}.

No existing response-adaptive design, however, targets an allocation driven jointly by several prioritized, possibly censored, endpoints of mixed type. This is a genuinely different generalization of the efficient randomized-adaptive design from the extension of \citet{AlkhnefrHuZhai2025} to multiple arms, or the further generalization by \citet{ZhangLX2026}: it generalizes along the axis of endpoint dimension and type rather than arm count, and it requires new asymptotic theory because the net benefit is a $U$-statistic rather than an average of independent and identically distributed summaries. The present paper develops this extension. Writing $\rho(\cdot)$ for the target allocation function, its argument, previously a vector of per-arm parameter estimates, is replaced by the vector of estimated pairwise net treatment benefits between all pairs of arms, computed from a hierarchy of prioritized endpoints. Because each pairwise net benefit is a two-sample $U$-statistic evaluated on a sequentially accruing, design-dependent sample size, its usual fixed-sample Hájek projection \citep{Hoeffding1948,vanderVaart1998} does not immediately apply; a sequential version of the projection, valid under the adaptive sampling scheme, is proved and used in place of the Bahadur-type representation on which the proof of \citet{HuZhangHe2009} rests. The martingale and stopping-time arguments of \citet{HuZhangHe2009} and \citet{AlkhnefrHuZhai2025} otherwise carry over with little change, and indeed are somewhat simplified by the fact that the pairwise comparison kernel underlying the net benefit is bounded.

Three contributions follow. First, the multi-arm efficient randomized-adaptive design is defined with a target driven by generalized pairwise comparisons, for any number of arms and any finite, ordered collection of endpoints of mixed type, with a martingale-based proof of its strong consistency, rate of convergence, and asymptotic normality. Second, the resulting design is shown to attain the semiparametric efficiency bound for the variance of the allocation proportion within the class of designs whose target is a smooth function of the net benefit; because the net benefit is a rank-based, distribution-free functional rather than a finite-dimensional parametric one, this is a genuinely different, and weaker, claim than the classical Cramér--Rao attainment shown for the parametric case by \citet{HuZhangHe2009} and \citet{AlkhnefrHuZhai2025}, and it is stated as such. Third, the design's finite-sample behaviour is examined through simulation calibrated to a real, three-arm, phase III confirmatory trial in melanoma \citep{DummerEtAl2018}, and the same trial is used to illustrate a redesign under the proposed procedure. Section~\ref{sec:framework} sets out the framework and the design; Section~\ref{sec:asymptotics} states the asymptotic properties, with proofs in the Appendix; Section~\ref{sec:simulation} reports the simulation study; Section~\ref{sec:application} redesigns the melanoma trial; Section~\ref{sec:regulatory} discusses the design in relation to current regulatory guidance on adaptive designs for confirmatory trials; and Section~\ref{sec:discussion} concludes.

\section{Framework}
\label{sec:framework}

\subsection{Multiple arms, prioritized endpoints, and generalized pairwise comparisons}

Consider a trial with $K\ (K\geq2)$ arms, to which $n$ patients are sequentially and without delay randomized. Let $X_{m,k}=1$ if the $m$th patient is assigned to arm $k$ and $0$ otherwise, and $N_{m,k}=\sum_{i=1}^m X_{i,k}$. Each patient has a response vector with $L\geq1$ components, prioritized in a clinically pre-specified order, $1\succ2\succ\cdots\succ L$; the components may be of mixed type, for example a time-to-event endpoint possibly subject to independent right censoring, followed by a binary endpoint, followed by a further binary or ordinal endpoint. Write $\xi_{i,k}=(\xi^{(1)}_{i,k},\ldots,\xi^{(L)}_{i,k})$ for the response of the $i$th patient assigned to arm $k$, and assume that $\{\xi_{i,k}\}_{i\geq1}$ is a sequence of independent and identically distributed random vectors, independent across $k$ and independent of the randomization mechanism; this exogeneity of potential responses is standard in the response-adaptive design literature \citep{HuZhangHe2009,RosenbergerLachin2016}.

For two patients on arms $j$ and $k$, with responses $a=(a^{(1)},\ldots,a^{(L)})$ and $b=(b^{(1)},\ldots,b^{(L)})$, the generalized pairwise comparison of \citet{Buyse2010} proceeds lexicographically: $a$ is compared with $b$ on endpoint 1; if the pair is a tie, or uninformative because of censoring, at endpoint 1, the comparison passes to endpoint 2; and so on. Writing $h(a,b)\in\{-1,0,1\}$ for the outcome of this cascade, $+1$ if $a$ wins, $-1$ if $b$ wins and $0$ if the pair is tied at every endpoint, $h$ is antisymmetric, $h(a,b)=-h(b,a)$, and bounded, $|h|\leq1$. The scoring rule at a time-to-event endpoint follows \citet{Gehan1965} for the uncensored case and the censoring correction of \citet{PeronEtAl2018} otherwise. The net treatment benefit of arm $j$ over arm $k$ is $\Delta_{jk}=E\{h(\xi_{1,j},\xi_{1,k})\}$, and after $m$ patients it is estimated by the two-sample $U$-statistic
\begin{equation}
\label{eq:netbenefit}
\hat\Delta_{m,jk}=\frac{1}{N_{m,j}N_{m,k}}\sum_{i=1}^{N_{m,j}}\sum_{i'=1}^{N_{m,k}}h(\xi_{i,j},\xi_{i',k}),
\end{equation}
computed over the patients assigned to arms $j$ and $k$ among the first $m$. Stacking the $K(K-1)/2$ pairwise comparisons gives $\hat{\boldsymbol\Delta}_m=(\hat\Delta_{m,jk})_{j<k}$, estimating $\boldsymbol\Delta=(\Delta_{jk})_{j<k}$.

\subsection{Generalized Bahadur representation for the net benefit}
\label{subsec:bahadur}

The representation on which \citet{HuZhangHe2009} build their theory is the Bahadur-type expansion of the per-arm parameter estimate, $\hat\theta_{m,k}=N_{m,k}^{-1}\sum_{j=1}^mX_{j,k}\xi_{j,k}+o(N_{m,k}^{-1/2})$, an average of independent and identically distributed terms. The net benefit \eqref{eq:netbenefit} is not such an average: it is a $U$-statistic, and moreover one evaluated on the design-dependent, sequentially accruing counts $N_{m,j}$ and $N_{m,k}$ rather than on fixed sample sizes. The following generalized Bahadur representation is established, proved in the Appendix by combining the classical two-sample Hájek projection with an argument that adaptive stopping, because it never depends on an as-yet-unrevealed response, preserves the conditional independence of the responses actually observed.

For $a$ in the support of arm $j$'s response, let $\varphi_{jk}^{(j)}(a)=E\{h(a,\xi_{1,k})\}-\Delta_{jk}$, and define $\varphi_{jk}^{(k)}$ symmetrically; both have mean zero. Then, under Condition~\ref{cond:A} of \S\ref{subsec:asymptotic}, as $m\to\infty$,
\begin{equation}
\label{eq:genBahadur}
\hat\Delta_{m,jk}-\Delta_{jk}=N_{m,j}^{-1}\sum_{i=1}^mX_{i,j}\varphi_{jk}^{(j)}(\xi_{i,j})+N_{m,k}^{-1}\sum_{i=1}^mX_{i,k}\varphi_{jk}^{(k)}(\xi_{i,k})+o_P(m^{-1/2}),
\end{equation}
uniformly over the $K(K-1)/2$ pairs $j<k$. Representation~\eqref{eq:genBahadur} plays exactly the role of eq.~(2.1) of \citet{HuZhangHe2009} in the proofs of \S\ref{sec:asymptotics}, with $\varphi_{jk}^{(k)}(\xi_{i,k})$ in place of $\xi_{i,k}-\theta_k$; because $h$ is bounded, the martingale increments built from $\varphi_{jk}^{(k)}$ are themselves bounded, which simplifies several of the maximal inequalities used by \citet{HuZhangHe2009} and \citet{AlkhnefrHuZhai2025}.

\subsection{Multi-arm efficient randomized-adaptive design with generalized pairwise comparisons}

The desired allocation proportion is a function $\rho(\cdot)=(\rho_1(\cdot),\ldots,\rho_K(\cdot))$ of the net benefit vector, $\rho:\mathbb{R}^{K(K-1)/2}\to(0,1)^K$, $\sum_k\rho_k(\cdot)=1$, chosen so that $N_{n,\cdot}/n\to\boldsymbol\upsilon=\rho(\boldsymbol\Delta)$. An example, used throughout \S\S\ref{sec:simulation}--\ref{sec:application}, is $\rho_k(\boldsymbol\Delta)=(1+\bar\Delta_k)/\sum_{k'=1}^K(1+\bar\Delta_{k'})$, where $\bar\Delta_k=(K-1)^{-1}\sum_{j\neq k}\Delta_{kj}$ is the mean net benefit of arm $k$ against the remaining arms; because $\bar\Delta_k\in[-1,1]$, the weights $1+\bar\Delta_k$ are automatically non-negative and the target is non-degenerate whenever no arm is uniformly dominated. Other targets, for example ones that penalize a low-priority toxicity endpoint more or less steeply, are readily accommodated by an appropriate choice of $\rho$.

The multi-arm generalized pairwise comparison efficient randomized-adaptive design proceeds as follows. An initial $Km_0$ patients are allocated by restricted randomization, $m_0$ to each arm. Given $m\ (m\geq Km_0)$ previously allocated patients and their observed responses, compute $\hat{\boldsymbol\Delta}_m$ from \eqref{eq:netbenefit} and $\hat\rho_m=\rho(\hat{\boldsymbol\Delta}_m)$. The $(m+1)$th patient is then allocated to arm $k$ with probability
\begin{equation}
\label{eq:design}
p_{m+1,k}=
\begin{cases}
\alpha\hat\rho_{m,k}, & N_{m,k}/m>\hat\rho_{m,k},\\
\hat\rho_{m,k}, & N_{m,k}/m=\hat\rho_{m,k},\\
\dfrac{(1-\alpha)\sum_{j\in S}\hat\rho_{m,j}}{|T|}+\hat\rho_{m,k}, & N_{m,k}/m<\hat\rho_{m,k},
\end{cases}
\end{equation}
where $S=\{j: N_{m,j}/m>\hat\rho_{m,j}\}$, $T=\{j: N_{m,j}/m<\hat\rho_{m,j}\}$, and $\alpha\in(0,1)$ controls the degree of randomization. Rule~\eqref{eq:design} is formally identical to the multi-arm efficient randomized-adaptive design of \citet{AlkhnefrHuZhai2025}, and reduces to the original two-arm design of \citet{HuZhangHe2009} when $K=2$; the difference, and the source of the new asymptotic theory in \S\ref{sec:asymptotics}, lies entirely in $\hat\rho_m$ now being computed from the generalized-pairwise-comparison estimator $\hat{\boldsymbol\Delta}_m$ rather than from per-arm parameter estimates.

\section{Asymptotic properties}
\label{sec:asymptotics}

\subsection{Conditions}
\label{subsec:asymptotic}

\begin{condition}
\label{cond:A}
The response sequences $\{\xi_{i,k}\}_{i\geq1}$, $k=1,\ldots,K$, are independent and identically distributed within arm, independent across arms and of the randomization mechanism; for time-to-event components, censoring is independent with survival function bounded away from zero on the support of interest; and the projections $\varphi_{jk}^{(j)}$, $\varphi_{jk}^{(k)}$ of \S\ref{subsec:bahadur} have finite variance for every pair $j<k$.
\end{condition}

\begin{condition}
\label{cond:B}
The map $y\mapsto\rho(y):\mathbb{R}^{K(K-1)/2}\to(0,1)^K$ is continuous and twice differentiable at $\boldsymbol\Delta$, with $\rho(\boldsymbol\Delta)=\boldsymbol\upsilon$.
\end{condition}

Condition~\ref{cond:A} is satisfied by the mixed-type endpoint hierarchies considered in \S\S\ref{sec:simulation}--\ref{sec:application}, because the pairwise comparison kernel $h$ is bounded and the underlying responses have well-defined marginal laws. Condition~\ref{cond:B} is the same smoothness requirement as Condition~B of \citet{HuZhangHe2009} and Condition~(B) of \citet{AlkhnefrHuZhai2025}, transferred from the parameter vector $\boldsymbol\Theta$ to the net benefit vector $\boldsymbol\Delta$; because \eqref{eq:genBahadur} shows $\hat{\boldsymbol\Delta}_m$ to be asymptotically linear, ordinary multivariate differentiability of $\rho$ suffices and no functional delta method is required.

Write $\boldsymbol\Sigma$ for the $K(K-1)/2\times K(K-1)/2$ asymptotic covariance matrix of $\sqrt m(\hat{\boldsymbol\Delta}_m-\boldsymbol\Delta)$, which by \eqref{eq:genBahadur} and the martingale central limit theorem (Appendix 1) has $(jk,j'k')$ entry built from $\cov\{\varphi_{jk}^{(\cdot)}(\xi),\varphi_{j'k'}^{(\cdot)}(\xi)\}$ over shared arms, and set
\begin{equation}
\label{eq:sigma2gpc}
\sigma^2_{\GPC}=\left(\frac{\partial\rho}{\partial y}\Big|_{\boldsymbol\Delta}\right)^{\!\T}\boldsymbol\Sigma\,\frac{\partial\rho}{\partial y}\Big|_{\boldsymbol\Delta}.
\end{equation}

\subsection{Main results}

\begin{theorem}
\label{thm:consistency}
Under Conditions~\ref{cond:A} and \ref{cond:B}, as $n\to\infty$,
\begin{equation}
\label{eq:thm1a}
|N_{n,k}-n\hat\rho_{n,k}|=o_P(n^{1/2}),\qquad |N_{n,k}-n\hat\rho_{n,k}|=O\{(n\log\log n)^{1/2}\}\quad\text{a.s.}\ (k=1,\ldots,K),
\end{equation}
and
\begin{equation}
\label{eq:thm1b}
\hat{\boldsymbol\Delta}_n-\boldsymbol\Delta=O\{(n^{-1}\log\log n)^{1/2}\}\ \text{a.s.},\qquad N_{n,k}-n\upsilon_k=O\{(n\log\log n)^{1/2}\}\ \text{a.s.}
\end{equation}
\end{theorem}

\begin{theorem}
\label{thm:normality}
Under Conditions~\ref{cond:A} and \ref{cond:B}, as $n\to\infty$,
\begin{equation}
\label{eq:thm2}
n^{1/2}(\hat{\boldsymbol\Delta}_n-\boldsymbol\Delta)\stackrel{D}{\to}N(0,\boldsymbol\Sigma),\qquad n^{1/2}(N_{n,k}/n-\upsilon_k)\stackrel{D}{\to}N(0,\sigma^2_{\GPC,k})\ (k=1,\ldots,K),
\end{equation}
with $\sigma^2_{\GPC,k}$ the $k$th diagonal-block quantity from \eqref{eq:sigma2gpc}.
\end{theorem}

\begin{theorem}
\label{thm:efficiency}
Under Conditions~\ref{cond:A} and \ref{cond:B}, $\sigma^2_{\GPC}$ in \eqref{eq:sigma2gpc} is the minimal achievable asymptotic variance of the allocation proportion among designs whose target is a smooth function of an asymptotically linear estimator of $\boldsymbol\Delta$; that is, the multi-arm generalized-pairwise-comparison efficient randomized-adaptive design attains the semiparametric efficiency bound for this class of designs.
\end{theorem}

Theorems~\ref{thm:consistency} and \ref{thm:normality} generalize Theorem~2.1 of \citet{HuZhangHe2009} and Theorem~1 of \citet{AlkhnefrHuZhai2025}; the proof follows the same stopping-time and martingale argument, with the substitutions detailed in Appendix 1. Theorem~\ref{thm:efficiency} is deliberately weaker than the classical Cramér--Rao attainment of Theorem~2.2 of \citet{HuZhangHe2009} and Theorem~2 of \citet{AlkhnefrHuZhai2025}: because $\boldsymbol\Delta$ is a rank-based, distribution-free functional rather than a smooth finite-dimensional parameter, there is in general no Fisher information for it, and the appropriate optimality criterion is semiparametric efficiency \citep{BickelEtAl1993} for the $U$-statistic estimator of $\boldsymbol\Delta$, combined with the fact, shown in Appendix 1, that no design tracking $\hat\rho_m$ at a slower rate than \eqref{eq:design} can improve on it. When $L=1$ and the single endpoint is such that the net benefit reduces to a monotone function of a scalar parameter difference, $\sigma^2_{\GPC}$ reduces to the classical variance formulae of \citet{HuZhangHe2009} and \citet{AlkhnefrHuZhai2025}; this reduction is not reproduced, since it is immediate from \eqref{eq:sigma2gpc}.

\begin{remark}[relation to the power--variability template of \citet{HuRosenberger2003}]
\label{rem:hurosenberger}
\citet{HuRosenberger2003} give a general template for comparing response-adaptive designs that share a common target allocation $\rho$: expanding the noncentrality parameter of the usual chi-squared or Wald statistic for treatment comparison in a Taylor series about $\rho$, they show that, to leading order, the power of the resulting test is a decreasing function of the asymptotic variance of the allocation proportion, so that among designs converging to the same $\rho$, the one with the smallest such variance is asymptotically most powerful. Theorem~\ref{thm:efficiency} identifies the multi-arm generalized-pairwise-comparison efficient randomized-adaptive design as attaining the smallest asymptotic variance $\sigma^2_{\GPC}$ among designs targeting $\rho(\boldsymbol\Delta)$; by the template of \citet{HuRosenberger2003}, it therefore also attains, to leading order, the greatest power among that class. This is a comparison within a fixed target: it does not by itself imply that the proposed design has greater power than a design targeting a \emph{different} allocation function, since a design's power depends on both the variance with which it tracks its target and on the target itself, and the two effects are not generally separable from the variance alone. Section~4.2 examines this distinction empirically, comparing pairs of designs that share an identical target and differ only in tracking mechanism.
\end{remark}

\subsection{Example: three arms, efficacy and safety}
\label{subsec:example}

Consider $K=3$ arms, and a two-tier hierarchy in which the first endpoint is a time-to-event outcome, censored, and the second is a binary safety indicator, unfavourable direction. Write $S_k(t)$ for the survival function on arm $k$ and $d_k$ for the probability of the safety event. If time-to-event responses are compared by the uncensored Gehan rule and mutual independence of the two endpoints is assumed given arm, so that $\Delta_{jk}=\Delta^{\rm TTE}_{jk}+\pi^{\rm tie}_{jk}\Delta^{\rm safety}_{jk}$, with $\pi^{\rm tie}_{jk}$ the probability that a pair is tied on the time-to-event endpoint, then $\varphi^{(j)}_{jk}$ decomposes into a sum of a term depending on arm $j$'s time-to-event value alone and a second term depending on both the time-to-event and safety values, and the two terms are not in general orthogonal even under within-patient independence, because the second term is itself modulated by the tie probability at the first endpoint. Consequently $\sigma^2_{\GPC}$ does not admit a closed form as compact as the moment-based examples of \citet{AlkhnefrHuZhai2025}; \S\ref{sec:simulation} evaluates it numerically under a parametric data-generating model and validates the evaluation against the empirical variance of $N_{n,k}/n$ over repeated simulation.

\subsection{Delayed responses}
\label{subsec:delayed}

The framework of \S\ref{sec:framework} assumed each patient's full endpoint vector is available before the next allocation decision. In practice this rarely holds exactly: an objective-response determination may await central radiological review, and even a time-to-event endpoint contributes only a censored value until the event, or a pre-specified analysis time, is reached. It is shown here that the asymptotic properties of \S\ref{subsec:asymptotic}--\S3.2 persist under delayed ascertainment of essentially the same generality as established for the original efficient randomized-adaptive design by \citet{BaiHuRosenberger2002}, and that, for the hierarchy of \S\ref{sec:framework}, delay can be accommodated by the same censoring-cascade mechanism already used for the time-to-event tier, rather than requiring separate machinery.

Formally, let $\xi_{i,k}$ be revealed to the allocation mechanism not at the time patient $i$ is assigned to arm $k$, but at a later time, after a delay $\tau_{i,k}\geq0$; write $\xi_{i,k}(m)$ for the information about $\xi_{i,k}$ available when the $(m+1)$th allocation decision is made, which equals $\xi_{i,k}$ if the delay has elapsed by then and is otherwise a coarsening of $\xi_{i,k}$ obtained by treating any endpoint not yet ascertained as administratively censored, or absent, at the current time, cascading to the next endpoint in the hierarchy exactly as in \S\ref{sec:framework}.

\begin{condition}
\label{cond:C}
The delays $\tau_{i,k}$ are independent and identically distributed within arm, independent of $\{\xi_{i,k}\}$ and of the randomization mechanism, with $E(\tau_{i,k})<\infty$.
\end{condition}

Condition~\ref{cond:C} is the direct analogue, for the present design, of the delay mechanism under which \citet{BaiHuRosenberger2002} established that the doubly adaptive biased coin design retains its asymptotic properties, and of the remark of \citet{HuZhangHe2009} that their Theorems~2.1 and 2.2 continue to hold for the original efficient randomized-adaptive design under the same mechanism.

\begin{proposition}
\label{prop:delay}
Under Conditions~\ref{cond:A}, \ref{cond:B} and \ref{cond:C}, the multi-arm generalized-pairwise-comparison efficient randomized-adaptive design, applied with $\hat{\boldsymbol\Delta}_m$ computed from $\{\xi_{i,j}(m),\xi_{i,k}(m)\}$ in place of $\{\xi_{i,j},\xi_{i,k}\}$ in \eqref{eq:netbenefit}, satisfies Theorems~\ref{thm:consistency}--\ref{thm:efficiency} unchanged.
\end{proposition}

\emph{Justification.} Two things must be checked: that the generalized Bahadur representation \eqref{eq:genBahadur} continues to hold with the coarsened data, and that the resulting variance is unaffected in the limit. For the first, write $\hat\Delta_{m,jk}(m)$ for the net benefit computed from $\{\xi_{i,\cdot}(m)\}$, and decompose $\hat\Delta_{m,jk}(m)-\hat\Delta_{m,jk}=\hat\Delta_{m,jk}(m)-\hat\Delta_{m,jk}$, the discrepancy contributed by patients whose relevant endpoint has not yet been ascertained; by Condition~\ref{cond:C}, the number of such patients among the first $m$ is, for each arm, a sum of $m$ independent and identically distributed indicators $\{\tau_{i,k}>m-i\}$ with mean bounded by $E(\tau_{i,k})$ uniformly in $m$ (an immediate consequence of $E(\tau_{i,k})<\infty$, since $\sum_{i=1}^m\pr(\tau_{i,k}>m-i)\leq\sum_{l=0}^{\infty}\pr(\tau_{1,k}>l)=E(\tau_{1,k})+1<\infty$), so the fraction of arm-$k$ patients with an unascertained top-priority endpoint at time $m$ is $O_P(1/m)=o_P(1)$. Because the pairwise comparison kernel $h$ is bounded, replacing an ascertained value by its coarsening changes the corresponding row or column of the kernel matrix in \eqref{eq:netbenefit} by at most $2$ in absolute value for each affected patient, so $|\hat\Delta_{m,jk}(m)-\hat\Delta_{m,jk}|=O_P(1/m)=o_P(m^{-1/2})$, uniformly over pairs by the same union-bound argument as Lemma~2 of Appendix~1. The generalized Bahadur representation \eqref{eq:genBahadur} therefore holds for $\hat{\boldsymbol\Delta}_m(m)$ with the same linear part and influence functions as for $\hat{\boldsymbol\Delta}_m$, the delay contributing only to the already-present $o_P(m^{-1/2})$ remainder. For the second, since the delay-induced discrepancy is $o_P(m^{-1/2})$, it does not affect the limiting covariance $\boldsymbol\Sigma$ of \eqref{eq:sigma2gpc}, and the martingale argument underlying Theorem~\ref{thm:consistency} (Appendix~1, \S{}A4) goes through with $X_{i,k}[\varphi_k\{\xi_{i,k}(m)\}-E\{\varphi_k(\xi_{i,k})\}]$ in place of $X_{i,k}[\varphi_k(\xi_{i,k})-E\{\varphi_k(\xi_{i,k})\}]$, which remains a bounded-increment martingale difference sequence with the same conditional mean zero, since $\xi_{i,k}(m)$ is, by Condition~\ref{cond:C}, independent of the allocation history given $\xi_{i,k}$. $\blacksquare$

Proposition~\ref{prop:delay} shows that delayed ascertainment need not be treated as a separate complication in the design's implementation: the same tie-breaking cascade of \S\ref{sec:framework} that carries a comparison from a censored time-to-event tier to the next endpoint in the hierarchy also carries it, formally identically, from an endpoint whose value has simply not yet arrived. Section~\ref{sec:simulation} demonstrates this numerically by delaying objective-response ascertainment relative to enrolment and confirming that the operating characteristics of \S4.2 are materially unaffected, consistent with the $o_P(m^{-1/2})$ rate established above.

\section{Simulation study}
\label{sec:simulation}

\subsection{Design of the study}
\label{subsec:simdesign}

Five allocation designs are compared by simulation, all built on the multi-arm \ERADE{} allocation rule \eqref{eq:design} with restricted randomization of $6$ patients per arm and $\alpha=1/2$, and differing only in the target $\rho(\cdot)$: the proposed design, with $\rho_k(\boldsymbol\Delta)=(1+\bar\Delta_k)/\sum_{k'=1}^K(1+\bar\Delta_{k'})$ of \S\ref{sec:framework}, $\bar\Delta_k=(K-1)^{-1}\sum_{j\neq k}\Delta_{kj}$; a response-rate-only design, targeting the \citet{RosenbergerEtAl2001} allocation $\rho_k(\mathbf p)=\sqrt{p_k}/\sum_{k'}\sqrt{p_{k'}}$, where $p_k$ is the objective-response probability on arm $k$, estimated by the continuity-corrected rate $\hat p_{m,k}=\{\sum_{i\leq m}X_{i,k}R_i+0\!\cdot\!5\}/(N_{m,k}+1)$ of \citet{HuZhangHe2009}, $R_i$ the objective-response indicator of the $i$th patient; a balanced design, targeting the fixed allocation $\rho=(1/K,\ldots,1/K)$ regardless of the data, included to isolate the contribution of the target function from the efficiency with which any target is tracked (\S4.2); and complete randomization, $p_{m+1,k}\equiv1/K$ for every $m$. A fifth design, the doubly adaptive biased coin design of \citet{HuZhang2004}, targets the same response-rate allocation $\rho(\hat{\mathbf p}_m)$ as the second design above, but with allocation probability $p_{m+1,k}\propto\hat\rho_{m,k}\{\hat\rho_{m,k}/(N_{m,k}/m)\}^\gamma$, $\gamma=2$, in place of \eqref{eq:design}.

Response vectors have three components, generated independently given arm in the base case (a further scenario induces correlation, \S\ref{subsec:sensitivity}): progression-free survival, Weibull with shape parameter $\kappa=1\!\cdot\!3$ and scale $\lambda_k=\mathrm{median}_k/(\log2)^{1/\kappa}$ chosen to match the arm-specific median of the calibration scenario, subject to independent administrative censoring at $36$ months unless stated otherwise; a Bernoulli objective-response indicator with arm-specific probability; and a Bernoulli toxicity indicator with arm-specific probability, in the unfavourable direction. The endpoint hierarchy is progression-free survival, then objective response, then toxicity, the priority order appropriate to a confirmatory, late-phase setting in which efficacy is weighted above safety. Comparisons at the time-to-event tier use the Gehan (1965) rule for uncensored pairs and the censoring-adjusted rule of \citet{PeronEtAl2018} otherwise (\S\ref{sec:framework}). Treatment differences are tested by the Wald statistic of Appendix 1, built from the estimated pairwise net benefits against arm $1$ and their $U$-statistic covariance, referred to $\chi^2_{K-1}$; its empirical size under the null scenario below, $5000$ replications, was $5\!\cdot\!5\%$, $5\!\cdot\!2\%$, $5\!\cdot\!7\%$, $5\!\cdot\!5\%$ and $5\!\cdot\!4\%$ for the five designs respectively, confirming the asymptotic approximation is adequate at the sample sizes used.

Four scenarios are considered, all with $K=3$ arms, corresponding to \textsc{combo}450, \textsc{enco}300 and \textsc{vem} of \S\ref{sec:application} in that order; Table~\ref{tab:scenarios} gives the arm-specific objective-response probability, toxicity probability and progression-free-survival median defining each. The base case is calibrated to the observed rates of the three-arm phase III \textsc{columbus} trial \citep{DummerEtAl2018,GogasEtAl2019}, described further in \S\ref{sec:application}, with sample size $n=577$ matching the real trial; the null scenario has identical arms, to assess the type~\textsc{i} error, $n=400$; the strong-separation scenario exaggerates the base-case gaps synthetically, to assess behaviour under a strong signal, $n=400$; and the moderate-separation scenario uses smaller, synthetic gaps together with a smaller sample size, $n=200$, chosen so that power is not saturated and differences between designs are informative for that purpose. Each scenario--design combination was run for $5000$ replications, except where otherwise noted for the sensitivity analyses of \S\ref{subsec:sensitivity}, which use $250$--$600$ replications, sufficient to detect a material change in behaviour at lower computational cost. All random variates were generated with \texttt{numpy}'s default random-number generator; the complete simulation code, sufficient with Table~\ref{tab:scenarios} to reproduce every number in this section exactly, is given in Appendix~3 and the Supplementary Material.

\begin{table}
\tbl{Data-generating parameters for the four simulation scenarios of \S\ref{sec:simulation}, by arm (\textsc{combo}450, \textsc{enco}300, \textsc{vem})}
{\begin{tabular}{lccc}
Scenario & Objective response & Toxicity & Progression-free survival median (months)\\
Base case & $0\cdot64,\ 0\cdot52,\ 0\cdot41$ & $0\cdot10,\ 0\cdot12,\ 0\cdot14$ & $14\cdot9,\ 9\cdot6,\ 7\cdot3$\\
Null & $0\cdot50,\ 0\cdot50,\ 0\cdot50$ & $0\cdot12,\ 0\cdot12,\ 0\cdot12$ & $10\cdot0,\ 10\cdot0,\ 10\cdot0$\\
Strong separation & $0\cdot70,\ 0\cdot45,\ 0\cdot25$ & $0\cdot08,\ 0\cdot18,\ 0\cdot30$ & $17\cdot0,\ 9\cdot0,\ 5\cdot0$\\
Moderate separation & $0\cdot55,\ 0\cdot47,\ 0\cdot40$ & $0\cdot10,\ 0\cdot13,\ 0\cdot17$ & $11\cdot5,\ 9\cdot0,\ 7\cdot5$\\
\end{tabular}}
\label{tab:scenarios}
\begin{tabnote}
Objective response and toxicity are Bernoulli probabilities; toxicity is in the unfavourable direction. The base case uses $n=577$ and the strong-separation scenario $n=400$, both with $36$-month administrative censoring for progression-free survival; the null scenario uses $n=400$; the moderate-separation scenario uses $n=200$, chosen so that power is not saturated.
\end{tabnote}
\end{table}

\subsection{Results}

Table~\ref{tab:main5000} reports the simulated allocation proportions, their standard deviations, and the resulting power and expected toxicity burden, calibrated to the real trial rates ($n=577$, the actual \textsc{columbus} sample size) and under the null. The proposed design concentrates substantially more allocation on the superior arm than either objective-response-only competitor, $46\!\cdot\!4\%$ against $37\!\cdot\!0\%$, while the type~\textsc{i} error of all five designs is close to the nominal $5\%$, with Monte Carlo standard error at most $0\!\cdot\!003$. The largest standard error for the allocation proportions in Table~\ref{tab:main5000} is $0\!\cdot\!0002$. The balanced efficient randomized-adaptive design, which targets the fixed allocation $\rho=(1/3,1/3,1/3)$ regardless of the observed data, is included to separate two distinct sources of the proposed design's advantage: its very low standard deviation in Table~\ref{tab:main5000} confirms that, for a constant target, the efficient randomized-adaptive mechanism itself achieves the zero asymptotic variance predicted by Theorem~\ref{thm:consistency} when $\rho(\cdot)$ does not depend on $n$, exactly as Efron's biased-coin design does for balancing two arms. Efron's design is itself the two-arm, balanced-target special case of \eqref{eq:design}, and is \emph{first-order efficient}: it attains the asymptotic lower bound on the variance of the allocation proportion for its target \citep{HuRosenbergerZhang2006}. Theorem~\ref{thm:efficiency} shows that the proposed design is first-order efficient in the corresponding semiparametric sense for every target $\rho(\boldsymbol\Delta)$, not merely the balanced one, so the balanced design here is a single, classically understood instance of a property the proposed design possesses generally. Yet its toxicity burden and power are indistinguishable from complete randomization, since a well-tracked but uninformative target confers no patient benefit. The gain reported for the proposed design therefore reflects the choice of target, not merely the efficiency with which any target is tracked.

\begin{table}
\tbl{Simulated allocation proportions (standard deviation), rejection rate and mean toxicity events for the base case, calibrated to the \textsc{columbus} trial ($n=577$), and under the null hypothesis of identical arms ($n=400$); 5000 replications per row}
{\begin{tabular}{lcccccc}
Scenario & Design & $N_{n,1}/n$ & $N_{n,2}/n$ & $N_{n,3}/n$ & Power & $E$(toxicity)\\
Base case & Proposed & 0$\cdot$464 (0$\cdot$020) & 0$\cdot$303 (0$\cdot$017) & 0$\cdot$233 (0$\cdot$016) & 1$\cdot$000 & 66$\cdot$7\\
 & Response-rate \ERADE & 0$\cdot$370 (0$\cdot$010) & 0$\cdot$334 (0$\cdot$010) & 0$\cdot$297 (0$\cdot$011) & 1$\cdot$000 & 68$\cdot$4\\
 & Response-rate \DBCD & 0$\cdot$370 (0$\cdot$014) & 0$\cdot$334 (0$\cdot$014) & 0$\cdot$296 (0$\cdot$015) & 1$\cdot$000 & 68$\cdot$2\\
 & Balanced \ERADE & 0$\cdot$333 (0$\cdot$003) & 0$\cdot$333 (0$\cdot$003) & 0$\cdot$333 (0$\cdot$003) & 1$\cdot$000 & 69$\cdot$4\\
 & Complete randomization & 0$\cdot$334 (0$\cdot$019) & 0$\cdot$333 (0$\cdot$019) & 0$\cdot$333 (0$\cdot$020) & 1$\cdot$000 & 69$\cdot$2\\
Null & Proposed & 0$\cdot$333 (0$\cdot$021) & 0$\cdot$333 (0$\cdot$021) & 0$\cdot$334 (0$\cdot$022) & 0$\cdot$055 & 48$\cdot$2\\
 & Response-rate \ERADE & 0$\cdot$333 (0$\cdot$013) & 0$\cdot$334 (0$\cdot$013) & 0$\cdot$333 (0$\cdot$013) & 0$\cdot$052 & 48$\cdot$2\\
 & Response-rate \DBCD & 0$\cdot$333 (0$\cdot$017) & 0$\cdot$333 (0$\cdot$017) & 0$\cdot$334 (0$\cdot$017) & 0$\cdot$057 & 48$\cdot$0\\
 & Balanced \ERADE & 0$\cdot$333 (0$\cdot$004) & 0$\cdot$333 (0$\cdot$004) & 0$\cdot$333 (0$\cdot$004) & 0$\cdot$055 & 48$\cdot$1\\
 & Complete randomization & 0$\cdot$333 (0$\cdot$023) & 0$\cdot$333 (0$\cdot$023) & 0$\cdot$334 (0$\cdot$023) & 0$\cdot$054 & 48$\cdot$0\\
\end{tabular}}
\label{tab:main5000}
\begin{tabnote}
\ERADE, efficient randomized-adaptive design; \DBCD, doubly adaptive biased coin design. $N_{n,k}/n$ is the simulated proportion of patients on arm $k$ at the end of the trial; arms 1, 2, 3 correspond to \textsc{combo}450, \textsc{enco}300 and \textsc{vem} in \S\ref{sec:application}. Balanced \ERADE targets the fixed allocation $(1/3,1/3,1/3)$, following the balanced-target design of \citet{AlkhnefrHuZhai2025}. Toxicity is the number of patients discontinuing treatment because of an adverse event. The largest Monte Carlo standard error of the power estimates is $0\cdot003$.
\end{tabnote}
\end{table}

Table~\ref{tab:power5000} compares power under the synthetic moderate-separation scenario, in which the sample size, $n=200$, is deliberately chosen so that power is not saturated, and reports the strong-separation scenario, in which allocation to the toxicity-heavy arm falls furthest under the proposed design. The proposed design attains slightly higher power than either objective-response-based competitor, the balanced design, and complete randomization, and allocates fewer patients to the least favourable arm under strong separation than any competitor except the balanced design, which allocates a full third of patients to the least favourable arm by construction and correspondingly shows the second-highest toxicity burden under strong separation.

\begin{table}
\tbl{Power and allocation to the least favourable arm under a moderate-separation scenario with $n=200$, and allocation to the least favourable arm and toxicity burden under a strong-separation scenario with $n=400$; 5000 replications per row}
{\begin{tabular}{lccc}
Design & Power ($n=200$) & $N_{n,3}/n$ (strong) & $E$(toxicity) (strong)\\
Proposed & 0$\cdot$722 (0$\cdot$006) & 0$\cdot$165 & 58$\cdot$3\\
Response-rate \ERADE & 0$\cdot$705 (0$\cdot$006) & 0$\cdot$250 & 67$\cdot$5\\
Response-rate \DBCD & 0$\cdot$716 (0$\cdot$006) & 0$\cdot$248 & 67$\cdot$2\\
Balanced \ERADE & 0$\cdot$704 (0$\cdot$006) & 0$\cdot$333 & 74$\cdot$5\\
Complete randomization & 0$\cdot$706 (0$\cdot$006) & 0$\cdot$333 & 74$\cdot$5\\
\end{tabular}}
\label{tab:power5000}
\begin{tabnote}
Standard errors of the power estimates are given in parentheses. $N_{n,3}/n$ is the simulated proportion allocated to the arm with the lowest net benefit.
\end{tabnote}
\end{table}

Figure~\ref{fig:trajectory} shows, for the base-case scenario, the simulated allocation proportion to each arm as a function of the number of patients enrolled, averaged over $500$ replications, for the proposed design and for the response-rate-only efficient randomized-adaptive design. The proposed design separates the arms earlier and more sharply, reflecting its use of the full endpoint hierarchy rather than the objective response rate alone.

\begin{figure}
\centering
\begin{tikzpicture}[scale=1]
  \draw[-{Latex[length=2mm]}] (0,0) -- (13.9,0) node[right, font=\footnotesize] {Patients enrolled};
  \draw[-{Latex[length=2mm]}] (0,0) -- (0,6.1) node[above, font=\footnotesize] {Allocation proportion};
  \foreach \xp/\lab in {0.0000/0,2.3570/100,4.7140/200,7.0711/300,9.4281/400,11.7851/500,13.6000/577} {
    \draw (\xp,0) -- (\xp,-0.08) node[below, font=\tiny] {\lab};
  }
  \foreach \yp/\lab in {0.6170/0.2,1.8511/0.3,3.0851/0.4,4.3191/0.5,5.5532/0.6} {
    \draw (0,\yp) -- (-0.08,\yp) node[left, font=\tiny] {\lab};
    \draw[gray!25] (0,\yp) -- (13.6,\yp);
  }
  \draw[thick, solid, color=black]
    plot[mark=*, mark size=1.1pt, mark options={fill=black}] coordinates
    {(1.1785,3.7009)(2.3570,3.8268)(3.5355,3.8218)(4.7140,3.8366)(7.0711,3.8564)(9.4281,3.8613)(11.7851,3.8687)(13.6000,3.8749)};
  \draw[thick, dashed, color=black]
    plot[mark=*, mark size=1.1pt, mark options={fill=black}] coordinates
    {(1.1785,2.6322)(2.3570,2.6618)(3.5355,2.6618)(4.7140,2.6680)(7.0711,2.6964)(9.4281,2.7026)(11.7851,2.7137)(13.6000,2.7124)};
  \draw[thick, solid, color=black]
    plot[mark=triangle*, mark size=1.5pt, mark options={fill=black}] coordinates
    {(1.1785,1.1834)(2.3570,1.0773)(3.5355,1.0736)(4.7140,1.0576)(7.0711,1.0415)(9.4281,1.0354)(11.7851,1.0317)(13.6000,1.0280)};
  \draw[thick, dashed, color=black]
    plot[mark=triangle*, mark size=1.5pt, mark options={fill=black}] coordinates
    {(1.1785,1.8967)(2.3570,1.8585)(3.5355,1.8683)(4.7140,1.8511)(7.0711,1.8289)(9.4281,1.8227)(11.7851,1.8042)(13.6000,1.8091)};
  \begin{scope}[shift={(1.2,6.7)}, font=\footnotesize]
    \draw[thick, solid, color=black] (0,0.30) -- (0.6,0.30) node[right] {Proposed, \textsc{combo}450};
    \draw[thick, dashed, color=black] (0,0) -- (0.6,0) node[right] {Response-rate \ERADE, \textsc{combo}450};
    \draw[thick, solid, color=black] (7.0,0.30) -- (7.6,0.30) node[right] {Proposed, \textsc{vem}};
    \draw[thick, dashed, color=black] (7.0,0) -- (7.6,0) node[right] {Response-rate \ERADE, \textsc{vem}};
  \end{scope}
\end{tikzpicture}
\caption{Simulated allocation proportion to the best arm, \textsc{combo}450 (encorafenib plus binimetinib), and the least favourable arm, \textsc{vem} (vemurafenib), as a function of the number of patients enrolled, base-case scenario, averaged over 500 replications. Solid lines, proposed design; dashed lines, response-rate-only efficient randomized-adaptive design; circular markers, \textsc{combo}450 (upper pair of curves); triangular markers, \textsc{vem} (lower pair of curves).}
\label{fig:trajectory}
\end{figure}
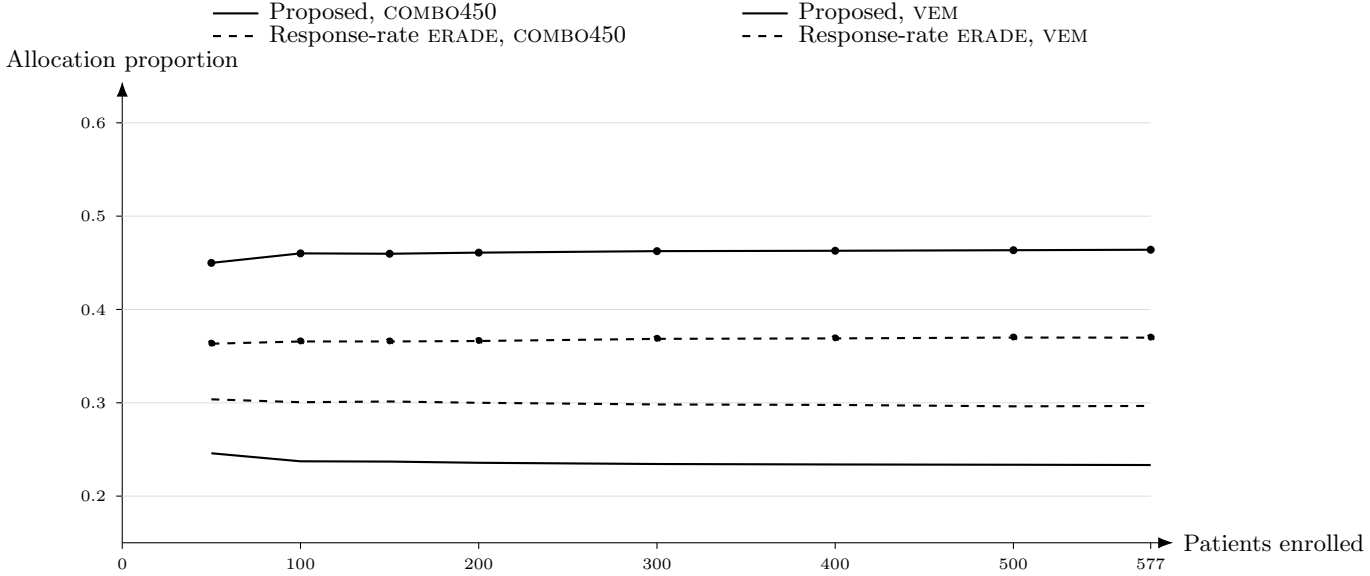

Figure~\ref{fig:barchart} summarizes, across the four scenarios of Tables~\ref{tab:main5000} and \ref{tab:power5000}, the mean number of toxicity events under each design, illustrating that the advantage of the proposed design grows with the separation between arms and vanishes, as it must, under the null.

\vspace{4pt}\noindent\textit{Power and the variability of the allocation proportion}\vspace{2pt}\par

Remark~\ref{rem:hurosenberger} connects the efficiency of Theorem~\ref{thm:efficiency} to the power--variability template of \citet{HuRosenberger2003}, but that template compares designs sharing a common target $\rho$; it is not, by itself, a statement about designs, such as those of Table~\ref{tab:power5000}, that target different allocation functions entirely. Figure~\ref{fig:powervar} plots the simulated power against the simulated allocation variance $n\var(N_{n,1}/n)$ for all five designs under the moderate-separation scenario. Read naively across all five points, the relationship is not monotone decreasing: the proposed design has both the largest variance and the highest power, because it targets an allocation that is more informative about the true ranking of arms than any competitor, and this effect of the target dominates the effect of tracking variability. Two pairs of designs in the figure do, however, share an identical target and isolate the comparison Remark~\ref{rem:hurosenberger} addresses, joined by a line in the figure. The balanced design and complete randomization both target $\rho=(1/3,1/3,1/3)$; their allocation variances differ by more than an order of magnitude, $0\!\cdot\!018$ against $0\!\cdot\!204$, yet their simulated power, $0\!\cdot\!704$ against $0\!\cdot\!706$, is indistinguishable given a Monte Carlo standard error of $0\!\cdot\!006$--$0\!\cdot\!007$. A plausible explanation is that the omnibus test of Appendix~1 estimates its own variance from the realized allocation in each replication rather than assuming the target is exactly attained, so that random deviations of $N_n/n$ from a fixed target are largely absorbed by the test's own variance estimate rather than costing power directly, a distinction between self-normalizing and plug-in test statistics that the classical template does not need to draw when, as is standard, the test statistic is built as if the target allocation were exactly known. The response-rate-only \ERADE{} and doubly adaptive biased coin designs, the second pair, also share an identical target, the \citet{RosenbergerEtAl2001} response-rate allocation, tracked by two different mechanisms; here the doubly adaptive biased coin design has both higher variance, $0\!\cdot\!125$ against $0\!\cdot\!089$, and marginally higher simulated power, $0\!\cdot\!716$ against $0\!\cdot\!705$, the opposite of the direction Remark~\ref{rem:hurosenberger} would suggest, though the gap is under two Monte Carlo standard errors and is not distinguishable from sampling variation at this replication count. Both comparisons are reported rather than only the one that is easier to interpret: the empirical support for Remark~\ref{rem:hurosenberger} in these simulations is that variability need not translate visibly into a power cost when a self-normalizing test is used, not a demonstration that lower variability mechanically buys power in every fixed-target comparison.

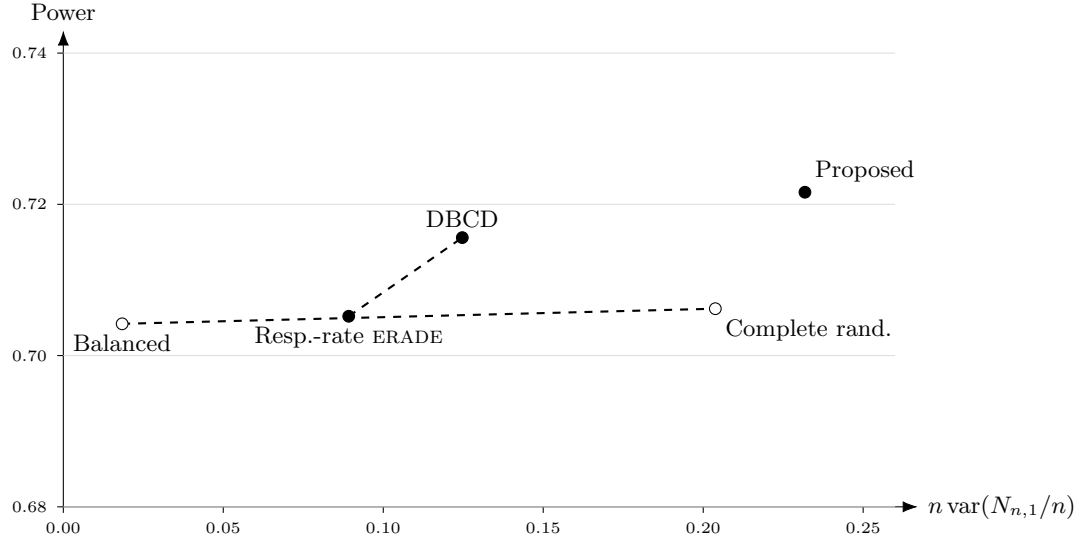
\begin{figure}
\centering
\begin{tikzpicture}
  \draw[-{Latex[length=2mm]}] (0,0) -- (11.3,0) node[right, font=\footnotesize] {$n\var(N_{n,1}/n)$};
  \draw[-{Latex[length=2mm]}] (0,0) -- (0,6.3) node[above, font=\footnotesize] {Power};
  \foreach \xp/\lab in {0.000/0.00,2.115/0.05,4.231/0.10,6.346/0.15,8.462/0.20,10.577/0.25} {
    \draw (\xp,0) -- (\xp,-0.08) node[below, font=\tiny] {\lab};
  }
  \foreach \yp/\lab in {0.000/0.68,2.000/0.70,4.000/0.72,6.000/0.74} {
    \draw (0,\yp) -- (-0.08,\yp) node[left, font=\tiny] {\lab};
    \draw[gray!25] (0,\yp) -- (11.0,\yp);
  }
  \draw[thick, dashed] (0.778,2.420) -- (8.622,2.620);
  \draw[thick, dashed] (3.774,2.520) -- (5.276,3.560);
  \filldraw[black] (9.807,4.160) circle (2.2pt) node[above right, font=\footnotesize] {Proposed};
  \filldraw[black] (3.774,2.520) circle (2.2pt) node[below, font=\footnotesize] {Resp.-rate \ERADE};
  \filldraw[black] (5.276,3.560) circle (2.2pt) node[above, font=\footnotesize] {DBCD};
  \draw[black, fill=white] (0.778,2.420) circle (2.2pt) node[below, font=\footnotesize] {Balanced};
  \draw[black, fill=white] (8.622,2.620) circle (2.2pt) node[below right, font=\footnotesize] {Complete rand.};
\end{tikzpicture}
\caption{Simulated power against simulated allocation variance $n\var(N_{n,1}/n)$, moderate-separation scenario, $n=200$, 5000 replications per design. Dashed lines connect the two pairs of designs that share an identical target allocation (balanced \ERADE{} with complete randomization; response-rate \ERADE{} with the doubly adaptive biased coin design), the only comparisons in the figure to which the power--variability template of \citet{HuRosenberger2003} applies directly.}
\label{fig:powervar}
\end{figure}

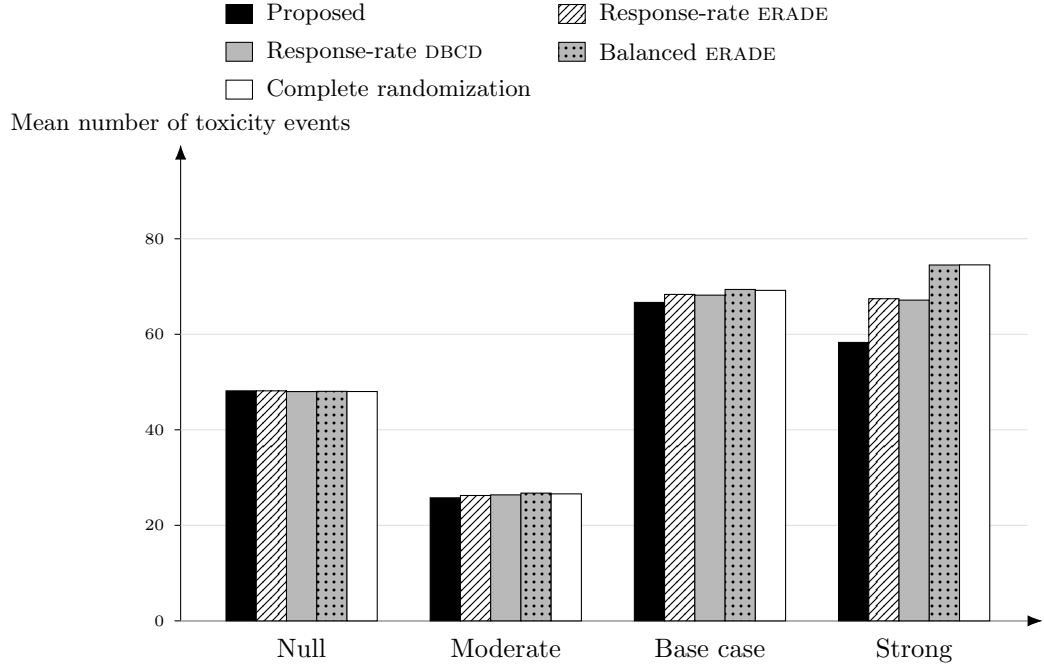
\begin{figure}
\centering
\begin{tikzpicture}
  \def\ymax{95}
  \def\ypix{6.0}
  \pgfmathsetmacro{\yscale}{\ypix/\ymax}
  \def\groupw{2.7}
  \def\barw{0.4}
  \draw[-{Latex[length=2mm]}] (0,0) -- (11.4,0);
  \draw[-{Latex[length=2mm]}] (0,0) -- (0,\ypix+0.3) node[above, font=\footnotesize] {Mean number of toxicity events};
  \foreach \y in {0,20,40,60,80} {
    \pgfmathsetmacro{\yp}{\y*\yscale}
    \draw (0,\yp) -- (-0.08,\yp) node[left, font=\tiny] {\y};
    \draw[gray!25] (0,\yp) -- (11.2,\yp);
  }
  \tikzset{
    fillProposed/.style={fill=black},
    fillERADE/.style={fill=white, postaction={pattern=north east lines}},
    fillDBCD/.style={fill=gray!55},
    fillBal/.style={fill=gray!55, postaction={pattern=dots}},
    fillCR/.style={fill=white},
  }
  \foreach \i/\name/\vProposed/\vERADE/\vDBCD/\vBal/\vCR in {
    0/Null/48.16/48.16/48.01/48.06/48.02,
    1/Moderate/25.76/26.25/26.38/26.75/26.59,
    2/{Base case}/66.68/68.35/68.20/69.39/69.20,
    3/Strong/58.30/67.45/67.16/74.51/74.54} {
    \pgfmathsetmacro{\gx}{\i*\groupw + 0.6}
    \pgfmathsetmacro{\hA}{\vProposed*\yscale}
    \pgfmathsetmacro{\hB}{\vERADE*\yscale}
    \pgfmathsetmacro{\hC}{\vDBCD*\yscale}
    \pgfmathsetmacro{\hE}{\vBal*\yscale}
    \pgfmathsetmacro{\hD}{\vCR*\yscale}
    \draw[fillProposed, draw=black] (\gx,0) rectangle ++(\barw,\hA);
    \draw[fillERADE, draw=black] (\gx+\barw,0) rectangle ++(\barw,\hB);
    \draw[fillDBCD, draw=black] (\gx+2*\barw,0) rectangle ++(\barw,\hC);
    \draw[fillBal, draw=black] (\gx+3*\barw,0) rectangle ++(\barw,\hE);
    \draw[fillCR, draw=black] (\gx+4*\barw,0) rectangle ++(\barw,\hD);
    \node[below, font=\small, yshift=-0.1cm] at (\gx+2.5*\barw,0) {\name};
  }
  \begin{scope}[shift={(0.6,\ypix+0.9)}, font=\footnotesize]
    \draw[fillProposed] (0,1.0) rectangle ++(0.35,0.25); \node[right] at (0.4,1.12) {Proposed};
    \draw[fillERADE] (4.4,1.0) rectangle ++(0.35,0.25); \node[right] at (4.8,1.12) {Response-rate \ERADE};
    \draw[fillDBCD] (0,0.5) rectangle ++(0.35,0.25); \node[right] at (0.4,0.62) {Response-rate \DBCD};
    \draw[fillBal] (4.4,0.5) rectangle ++(0.35,0.25); \node[right] at (4.8,0.62) {Balanced \ERADE};
    \draw[fillCR, draw=black] (0,0) rectangle ++(0.35,0.25); \node[right] at (0.4,0.12) {Complete randomization};
  \end{scope}
\end{tikzpicture}
\caption{Mean number of toxicity events (treatment discontinuations attributed to an adverse event) by design and scenario, 5000 replications per bar. Null and Base case use $n=400$ and $n=577$ respectively; Moderate and Strong are the synthetic separation scenarios of Table~\ref{tab:power5000} with $n=200$ and $n=400$.}
\label{fig:barchart}
\end{figure}

\subsection{Sensitivity analyses}
\label{subsec:sensitivity}

The robustness of these conclusions to four departures from the base-case data-generating model was examined, each run for a smaller number of replications, $250$--$600$, sufficient to detect a material change in behaviour. Varying the administrative follow-up duration between $12$ and $48$ months changed the allocation to the superior arm and the toxicity burden by less than one percentage point under either design. Replacing the Weibull progression-free survival distribution with a log-normal distribution of the same medians left the qualitative comparison between designs unchanged, as expected since the pairwise comparison kernel is rank-based rather than tied to a parametric family. Inducing a clinically plausible positive correlation, through a shared latent-health probit link, between longer progression-free survival, response and freedom from toxicity did not alter the direction or approximate size of the advantage of the proposed design over the response-rate-only designs.

A fourth analysis examines patient accrual rate, which the previous three analyses, generating each patient's response in full at the moment of allocation, could not address. A calendar-time version of the simulator was implemented in which patients arrive as a homogeneous Poisson process at a specified rate, and the allocation rule at each decision point uses only the progression-free-survival status actually available as of that calendar time for every previously enrolled patient, namely $\min\{$true progression-free survival, calendar time elapsed since that patient's own enrolment, $36$ months$\}$, censored if the true value exceeds this; the objective-response and toxicity indicators, which in practice are ascertained comparatively quickly, are taken to be available immediately. This is precisely the coarsening mechanism of \S\ref{subsec:delayed}, applied here to the naturally occurring immaturity of a time-to-event endpoint rather than to a review or adjudication delay, and Proposition~\ref{prop:delay} guarantees that the design's asymptotic properties are unaffected by it, since the fraction of immature comparisons at any fixed point vanishes as enrolment proceeds. Final analysis in every case uses fully-matured data, each patient followed for the full $36$-month administrative window, so that accrual rate affects only the information available to the adaptive allocation rule while it is running, not the comparability of the final reported operating characteristics.

Table~\ref{tab:accrual} reports the result at $n=300$, $400$ replications per cell, for slow ($10$ patients per month), base-case ($25$ patients per month, comparable to the roughly two-year accrual period of the real trial for $577$ patients), and fast ($60$ patients per month) accrual. The proposed design's concentration on the superior arm decreases modestly and monotonically as accrual speeds up, from $43\!\cdot\!3\%$ under slow accrual to $40\!\cdot\!6\%$ under fast accrual, because faster accrual means more of the patients contributing to any given allocation decision have had less calendar time to accrue mature progression-free-survival information, so a larger share of pairwise comparisons are decided at the lower-priority response or toxicity tiers rather than at progression-free survival. The response-rate-only design, which does not use progression-free survival at all, is exactly invariant to accrual rate, as expected. Power and toxicity burden are materially unaffected by accrual rate under either design, confirming that the modest loss of concentration under fast accrual does not translate into a loss of validity.

\begin{table}
\tbl{Sensitivity to patient accrual rate, base-case rates, $n=300$, 400 replications per cell}
{\begin{tabular}{lcccc}
Accrual & Design & $N_{n,1}/n$ & Power & $E$(toxicity)\\
Slow (10/month) & Proposed & 0$\cdot$433 & 1$\cdot$000 & 34$\cdot$8\\
 & Response-rate \ERADE & 0$\cdot$369 & 1$\cdot$000 & 35$\cdot$3\\
Base (25/month) & Proposed & 0$\cdot$416 & 0$\cdot$998 & 35$\cdot$2\\
 & Response-rate \ERADE & 0$\cdot$369 & 1$\cdot$000 & 35$\cdot$5\\
Fast (60/month) & Proposed & 0$\cdot$406 & 1$\cdot$000 & 34$\cdot$9\\
 & Response-rate \ERADE & 0$\cdot$369 & 1$\cdot$000 & 35$\cdot$5\\
\end{tabular}}
\label{tab:accrual}
\begin{tabnote}
$N_{n,1}/n$ is the simulated proportion allocated to \textsc{combo}450, the superior arm, at final (fully-matured) analysis. Accrual rate is patients enrolled per month, modelled as a homogeneous Poisson process.
\end{tabnote}
\end{table}

A fifth analysis demonstrates the delayed-response mechanism of \S\ref{subsec:delayed} directly, using a form of delay distinct from the accrual-driven progression-free-survival immaturity of Table~\ref{tab:accrual}: objective-response ascertainment, which in practice awaits scheduled imaging and central radiological adjudication, is delayed relative to a patient's enrolment by an independent $\mathrm{Exponential}$ adjudication time, while progression-free survival follows the usual fixed administrative window and toxicity is assumed to be recorded promptly. At each allocation decision, a patient whose objective-response adjudication has not yet returned contributes the sentinel state \emph{not yet ascertained} rather than a value of $0$ or $1$; the pairwise comparison kernel treats this exactly as a further tie, cascading to toxicity as it would for a tie at any other tier, and to nothing for the response-rate-only competitor's running rate estimate, which uses only patients whose response has resolved. Table~\ref{tab:orrdelay} reports the result at $n=300$, $400$ replications per cell, for a negligible ($0\!\cdot\!01$-month mean), moderate ($1\!\cdot\!5$-month mean) and substantial ($4$-month mean) adjudication delay. The effect on the proposed design is in the same direction as, but visibly smaller than, the effect of accrual rate in Table~\ref{tab:accrual}, because progression-free survival is already the top-priority endpoint: most pairwise comparisons are decided before the response tier is reached at all, so a delay confined to that tier has less to act on. Power and toxicity are again materially unaffected, consistent with Proposition~\ref{prop:delay}.

\begin{table}
\tbl{Sensitivity to delayed objective-response ascertainment, base-case rates, $n=300$, 400 replications per cell}
{\begin{tabular}{lcccc}
Adjudication delay & Design & $N_{n,1}/n$ & Power & $E$(toxicity)\\
Negligible (mean 0$\cdot$01 mo.) & Proposed & 0$\cdot$416 & 1$\cdot$000 & 35$\cdot$6\\
 & Response-rate \ERADE & 0$\cdot$369 & 1$\cdot$000 & 35$\cdot$7\\
Moderate (mean 1$\cdot$5 mo.) & Proposed & 0$\cdot$402 & 1$\cdot$000 & 35$\cdot$3\\
 & Response-rate \ERADE & 0$\cdot$370 & 1$\cdot$000 & 35$\cdot$7\\
Substantial (mean 4 mo.) & Proposed & 0$\cdot$390 & 1$\cdot$000 & 35$\cdot$4\\
 & Response-rate \ERADE & 0$\cdot$369 & 1$\cdot$000 & 35$\cdot$0\\
\end{tabular}}
\label{tab:orrdelay}
\begin{tabnote}
$N_{n,1}/n$ is the simulated proportion allocated to \textsc{combo}450 at final analysis, by which time every patient's objective response has resolved. Progression-free survival and toxicity ascertainment are as in the base case.
\end{tabnote}
\end{table}

Full results for all five sensitivity analyses are given in the Supplementary Material.

\section{Application: redesign of a confirmatory melanoma trial}
\label{sec:application}

\textsc{columbus} was a multicentre, open-label, phase III trial in patients with locally advanced, unresectable or metastatic \textit{BRAF}\,V600-mutant melanoma, randomized $1:1:1$ to encorafenib 450\,mg once daily plus binimetinib 45\,mg twice daily, encorafenib 300\,mg once daily alone, or vemurafenib 960\,mg twice daily, with progression-free survival as the pre-specified primary endpoint \citep{DummerEtAl2018}. It supported regulatory approval of the combination and is exactly the setting for which the present design is intended: three concurrently randomized arms, a clear priority ordering of efficacy over safety, and a genuine trade-off in only one of the two monotherapy arms, encorafenib alone having a materially worse toxicity profile than vemurafenib despite better efficacy. Table~\ref{tab:columbus} summarizes the observed rates used to calibrate the base-case scenario of \S\ref{sec:simulation}: objective response by blinded independent central review, updated median progression-free survival, and the proportion of patients discontinuing treatment because of an adverse event \citep{DummerEtAl2018,GogasEtAl2019}.

\begin{table}
\tbl{Observed outcome rates in the \textsc{columbus} trial, part 1 ($n=577$), used to calibrate the base-case simulation scenario}
{\begin{tabular}{lccc}
 & \textsc{combo}450 & \textsc{enco}300 & \textsc{vem}\\
Randomized, $n$ & 192 & 194 & 191\\
Objective response & 0$\cdot$64 & 0$\cdot$52 & 0$\cdot$41\\
Median progression-free survival, months & 14$\cdot$9 & 9$\cdot$6 & 7$\cdot$3\\
Discontinuation for adverse event & 0$\cdot$10 & 0$\cdot$12 & 0$\cdot$14\\
\end{tabular}}
\label{tab:columbus}
\begin{tabnote}
\textsc{combo}450, encorafenib 450\,mg plus binimetinib 45\,mg; \textsc{enco}300, encorafenib 300\,mg alone; \textsc{vem}, vemurafenib. Source: \citet{DummerEtAl2018}, updated analysis, and \citet{GogasEtAl2019}.
\end{tabnote}
\end{table}

Under the actual $1:1:1$ randomization used in \textsc{columbus}, $192$, $194$ and $191$ patients respectively were allocated to the three arms. Had the proposed design been used instead, with the same total sample size, the base-case row of Table~\ref{tab:main5000} shows that, on average, $268$, $175$ and $134$ patients would have been allocated to \textsc{combo}450, \textsc{enco}300 and \textsc{vem}, compared with $213$, $192$ and $171$ under the response-rate-only efficient randomized-adaptive design and $192$ under each arm as randomized. The redesign concentrates more patients on the treatment subsequently shown to be superior on every endpoint, at an expected cost of about two to three fewer patients experiencing a treatment-limiting adverse event than under the response-rate-only designs, and about two to three fewer than under the trial as actually conducted; the type~\textsc{i} error is unaffected, and power, already close to one at this sample size and effect, is unaffected in practice. The benefit is modest in absolute terms, reflecting the fact that \textsc{columbus}'s three arms did not differ enormously in safety; \S\ref{sec:simulation} shows that the same mechanism produces a substantially larger benefit when one arm is disproportionately toxic, a configuration common in dose-finding and dose-optimization settings.

\section{Relation to regulatory guidance on adaptive designs}
\label{sec:regulatory}

Two guidance documents bear directly on how a design of the kind proposed here would be viewed by regulators. The International Council for Harmonisation's E20 guideline on adaptive designs for clinical trials \citep{ICH_E20_2022} distinguishes designs by the extent to which the adaptation rule and its operating characteristics can be pre-specified and verified by simulation before the trial begins, and by whether the statistical validity of the final analysis, in particular control of the type~\textsc{i} error, can be established independently of the particular realization of the adaptation. The design of \S\ref{sec:framework} satisfies both requirements in the sense intended by \textsc{ich} E20: the allocation rule \eqref{eq:design} is fully specified in advance, given the endpoint hierarchy and the target function $\rho$, and Theorems~\ref{thm:consistency} and \ref{thm:normality} establish, without reference to any particular data realization, the asymptotic distribution needed to construct a valid test, which the simulations of \S\ref{sec:simulation} confirm holds its nominal size at realistic sample sizes. \textsc{ich} E20 places particular emphasis on response-adaptive randomization as a design meriting careful simulation-based justification of its operating characteristics prior to submission, precisely the exercise undertaken in \S\ref{sec:simulation}.

The United States Food and Drug Administration's guidance on adaptive designs for clinical trials \citep{FDA_adaptive_2019} similarly asks that the adaptation algorithm be completely pre-specified, that its statistical properties be characterized by simulation across a range of plausible scenarios, including the null, and that any change to the probability of a type~\textsc{i} error be quantified. It also draws a distinction, relevant here, between response-adaptive randomization that changes the allocation ratio based on interim outcomes and adaptations to sample size, arms, or endpoints; the present design is of the former kind, and does not itself alter the set of hypotheses being tested or drop arms, so many of the more delicate issues addressed by that guidance for arm-dropping or population-enrichment designs, such as the interpretability of a per-arm treatment effect after selection, do not arise. Both documents also emphasize the operational challenge of maintaining blinding and rapid access to accruing multivariate outcome data when the adaptation rule depends on more than a single, readily ascertained endpoint; this is a genuine practical requirement of the proposed design, since the allocation rule \eqref{eq:design} needs, at each update, the prioritized endpoint vector for every previously enrolled patient, including any time-to-event component, however immature. Sponsors intending to use the design would need data infrastructure supporting near-real-time capture and adjudication of the full endpoint hierarchy, not merely the top-priority endpoint, a requirement already common in trials using generalized pairwise comparisons or the win ratio for their primary analysis \citep{PeronEtAl2016,DeBackerEtAl2024}.

\section{Guidance for a statistical analysis plan}
\label{sec:sap}

Implementing the design in a confirmatory trial requires four things to be fixed in the statistical analysis plan before unblinding, agreed with the treating clinicians or drawn from historical or pilot data rather than chosen post hoc. First, the endpoint hierarchy itself: which endpoint is primary or co-primary and which are key secondary, and the clinically meaningful comparison rule at each tier, for example a minimally important difference for a continuous endpoint or the censoring-adjusted time-to-event rule of \S\ref{sec:framework}, exactly the hierarchy a sponsor must already specify for the analysis of co-primary or primary-and-key-secondary endpoints under current multiplicity guidance, so no new clinical judgement is demanded beyond what confirmatory trials with such endpoints already require. Second, the target function $\rho(\cdot)$ and the randomization constant $\alpha$, together with the burn-in size and the frequency with which the target is re-estimated (\S\ref{subsec:simdesign}), all of which are algorithmic choices fixed in advance rather than estimated from the trial's own data. Third, the anticipated nuisance parameters, arm-specific event rates, response probabilities and, for a time-to-event tier, the assumed survival distribution, elicited from clinicians or taken from historical or pilot-trial data, exactly as Table~\ref{tab:scenarios} was populated from \textsc{columbus}; and fourth, whether within-patient dependence among the endpoints is expected, since although Theorem~\ref{thm:consistency} does not require independence, the sensitivity analysis of \S\ref{subsec:sensitivity} is the template for checking that anticipated correlation does not materially alter the operating characteristics before the trial begins.

Sample size determination differs from a trial using complete randomization in kind, not merely in degree. Under complete randomization the allocation is fixed and known in advance, so power can be obtained from a closed-form noncentrality-parameter calculation for the planned test. Under the proposed design the realized allocation is itself random, so no such closed form exists in general; an approximate sample size can be obtained from Theorem~\ref{thm:normality} by substituting the anticipated $\boldsymbol\Delta$ and evaluating $\sigma^2_{\GPC}$ under the planned nuisance parameters, giving a normal-approximation calculation analogous to a Wald-test sample size, but because this is a first-order asymptotic approximation, and because \textsc{ich} E20 and the \textsc{fda} guidance of \S\ref{sec:regulatory} expect the operating characteristics of a response-adaptive design to be verified by simulation before submission, the recommended practice is to treat this calculation as a starting value and refine it by Monte Carlo simulation across a grid of sample sizes under the pre-specified scenario, exactly as in \S\ref{sec:simulation}, for which the code of Appendix~3 and the Supplementary Material provides a direct template.

\section{Discussion}
\label{sec:discussion}

The efficient randomized-adaptive design of \citet{HuZhangHe2009}, and its multi-arm generalization by \citet{AlkhnefrHuZhai2025}, has been extended to targets driven by the net treatment benefit of generalized pairwise comparisons across a prioritized hierarchy of endpoints of mixed type. The extension required a new asymptotic argument, a sequential Hájek projection for the net benefit under adaptive, design-dependent sampling, combined with the martingale technique underlying the original design; the resulting design attains the semiparametric efficiency bound for its class, a genuinely weaker but still substantive claim than the parametric Cramér--Rao attainment of the designs it generalizes. In simulation calibrated to a real phase III confirmatory trial, the design allocates more patients to the treatment ultimately shown to be superior, and fewer to less favourable treatments, than designs driven by a single endpoint, without affecting the type~\textsc{i} error, and this advantage grows with the degree to which arms differ in endpoints other than the top-priority one.

The present development leaves several questions open. The efficiency claim of Theorem~\ref{thm:efficiency} is relative to designs targeting a smooth function of the net benefit; whether a design targeting some other functional of the full multivariate response distribution could do better is not addressed here and would require a different, and possibly considerably harder, semiparametric efficiency bound. The endpoint hierarchy has been treated as fixed and known; incorporating covariate information into either the endpoint comparisons themselves, along the lines of covariate-adjusted response-adaptive designs for a single survival endpoint \citep{MukherjeeEtAl2023,MukherjeeEtAl2025}, including under competing risks \citep{MukherjeeJana2024}, or into the allocation rule directly, is a natural next step given the importance regulators place on treatment effect estimation within covariate strata. The scoring rule for the time-to-event tier has been taken to be that of \citet{Gehan1965} and \citet{PeronEtAl2018}; other choices, such as a restricted mean survival time contrast, would change the influence function of \S\ref{subsec:bahadur} but not the overall argument. Finally, the design has been developed for superiority comparisons; adapting it to non-inferiority or equivalence hypotheses on the net benefit scale, relevant to the de-escalation trials in which generalized pairwise comparisons have already found use \citep{DeBackerEtAl2024}, is left for future work.

\section*{Acknowledgement}

The author thanks the Population Health Sciences Institute, Newcastle University, for the resources supporting this research, and the funder acknowledged below for financial support.

\section*{Funding}

This research work is supported by Dr.\ Ayon Mukherjee's National Institute for Health and Care Research Grant [grant number: NIHR301614].

\section*{ORCID}
Ayon Mukherjee \orcidlink{0000-0002-9655-2605} \url{https://orcid.org/0000-0001-5461-1737} \\

\section*{Supplementary material}

The Supplementary Material includes full results of the sensitivity analyses of \S\ref{subsec:sensitivity}, additional simulation scenarios, and the complete simulation code with instructions for reproducing every table and figure in this paper.

\section*{Data availability statement}

No new clinical data were collected for this article. The real-data application of \S\ref{sec:application} uses only summary statistics already published by \citet{DummerEtAl2018} and \citet{GogasEtAl2019}; no individual patient data were accessed or are required to reproduce any result in this paper. All simulated data underlying \S\ref{sec:simulation} and \S\ref{sec:application} were generated by the code given in Appendix~3 and the Supplementary Material, which is sufficient to regenerate every dataset, table and figure in this article.

\appendix

\appendixone
\section*{Appendix 1}
\subsection*{Proofs}
\label{app:proofs}

Throughout, $\mathcal F_m=\sigma(X_{1,\cdot},\ldots,X_{m,\cdot},\text{responses through patient }m)$, and all limits are as $m\to\infty$ or $n\to\infty$ unless stated otherwise.

\vspace{4pt}\noindent\textit{A1. Exact two-sample Hájek projection}\vspace{2pt}\par

For fixed $n_1,n_2$ and independent samples $X_{1:n_1}$, $Y_{1:n_2}$ from the laws of arms $j,k$, write $U(n_1,n_2)$ for the $U$-statistic \eqref{eq:netbenefit} evaluated on these samples, $\Delta=E\,h(X,Y)$, $\varphi_1(x)=E_Y\,h(x,Y)-\Delta$, $\varphi_2(y)=E_X\,h(X,y)-\Delta$, $\sigma_1^2=\var\{\varphi_1(X)\}$, $\sigma_2^2=\var\{\varphi_2(Y)\}$, $\sigma_h^2=\var\{h(X,Y)\}$, and
\[
L(n_1,n_2)=n_1^{-1}\sum\varphi_1(X_i)+n_2^{-1}\sum\varphi_2(Y_i),\qquad R(n_1,n_2)=U(n_1,n_2)-\Delta-L(n_1,n_2).
\]

\begin{lemma}
\label{lem:A1}
For every $n_1,n_2\geq1$, $\var\{R(n_1,n_2)\}=(\sigma_h^2-\sigma_1^2-\sigma_2^2)/(n_1n_2)\leq1/(n_1n_2)$.
\end{lemma}

\begin{proof}
Expanding $\var\{U(n_1,n_2)\}$ over the four index-coincidence patterns of the double sum gives the exact identity $\var\{U(n_1,n_2)\}=\sigma_1^2/n_1+\sigma_2^2/n_2+(\sigma_h^2-\sigma_1^2-\sigma_2^2)/(n_1n_2)$. Since $\var\{L(n_1,n_2)\}=\sigma_1^2/n_1+\sigma_2^2/n_2$ exactly, and $L$ is the $L^2$-projection of $U-\Delta$ onto sums of separable one-sample functions \citep[][Lemma~12.3]{vanderVaart1998}, $\cov(L,R)=0$ and the stated identity for $\var\{R(n_1,n_2)\}$ follows. Boundedness of $h$ gives $\sigma_h^2\leq1$, $\sigma_1^2,\sigma_2^2\geq0$, hence the bound.
\end{proof}

\vspace{4pt}\noindent\textit{A2. Sequential preservation of conditional independence}\vspace{2pt}\par

Construct, for each arm $k$, an exogenous potential-response reservoir $\eta^{(k)}_1,\eta^{(k)}_2,\ldots$, independent and identically distributed, independent across $k$ and of all randomization; the response of the $l$th patient assigned to arm $k$ is $\eta^{(k)}_l$.

\begin{lemma}
\label{lem:A2}
For fixed $n=(n_1,\ldots,n_K)$ with $\sum_kn_k=m$, conditionally on $\{N_{m,\cdot}=n\}$, the vectors $\eta^{(1)}_{1:n_1},\ldots,\eta^{(K)}_{1:n_K}$ are independent, each independent and identically distributed.
\end{lemma}

\begin{proof}
By induction on $m$. The claim holds at $m=Km_0$ by construction. Given it holds at $m$, the event $\{N_{m+1,\cdot}=n+e_k\}=\{N_{m,\cdot}=n\}\cap\{X_{m+1,k}=1\}$ is, by \eqref{eq:design}, a function of $\mathcal F_m$ and independent auxiliary randomness, hence of $\eta^{(k)}_{1:n_k}$ and of the other arms' already-used values only, never of $\eta^{(k)}_{n_k+1}$; the newly revealed value $\eta^{(k)}_{n_k+1}$ is therefore independent of this event, and the inductive hypothesis extends. This is the discrete-time analogue of the classical fact that a stopping rule which never inspects the next unrevealed value of an independent and identically distributed sequence preserves the conditional independence of what has been revealed \citep[cf.][ch.~1--2]{Gut2009}.
\end{proof}

\vspace{4pt}\noindent\textit{A3. Proof of the generalized Bahadur representation \eqref{eq:genBahadur}}\vspace{2pt}\par

Fix $j<k$. By Lemma~\ref{lem:A2}, conditioning on $N_{m,\cdot}=n$ recovers the fixed-sample setting of Lemma~\ref{lem:A1}, so $E[R_{jk}(N_{m,j},N_{m,k})^2\mid N_{m,\cdot}=n]\leq1/(n_jn_k)$. Because $U$-statistics with bounded kernel satisfy the strong law of large numbers as the sample size diverges regardless of rate \citep{Hoeffding1961}, and $N_{m,j},N_{m,k}\to\infty$ a.s.\ whenever $\upsilon_j,\upsilon_k>0$ (established first, exactly as in eqs.~(A.6)--(A.9) of \citealp{HuZhangHe2009}, using only continuity of $\rho$ and the martingale strong law for the allocation counts, and not requiring \eqref{eq:genBahadur}), $\hat\Delta_{m,jk}\to\Delta_{jk}$ a.s., hence $N_{m,k}/m\to\upsilon_k$ a.s.\ for every $k$. Fixing $\varepsilon\in(0,\min_k\upsilon_k/2)$, on $\mathcal A_m=\{N_{m,j}\geq\varepsilon m,\,N_{m,k}\geq\varepsilon m\}$, $P(\mathcal A_m^c)\to0$ exponentially fast by Azuma--Hoeffding applied to the bounded-increment counting process $N_{m,k}$, so $E\{R_{jk}^2\mathbf 1_{\mathcal A_m}\}\leq1/(\varepsilon^2m^2)$ while, since $|R_{jk}|$ is bounded, $E\{R_{jk}^2\mathbf 1_{\mathcal A_m^c}\}=o(1/m)$. Hence $E\{R_{jk}(N_{m,j},N_{m,k})^2\}=O(1/m^2)$, giving $R_{jk}=O_P(1/m)=o_P(m^{-1/2})$ by Markov's inequality, and $R_{jk}=O\{(\log\log m)^{1/2}/m\}$ a.s.\ by Borel--Cantelli using the exponential tail bound. Uniformity over the $K(K-1)/2$ pairs follows by a union bound.

\vspace{4pt}\noindent\textit{A4. Proof of Theorems~\ref{thm:consistency} and \ref{thm:normality}}\vspace{2pt}\par

Fix an arm $k\in\{1,\ldots,K\}$; the argument is symmetric in $k$, and since $\sum_kN_{n,k}=n$ it suffices to establish \eqref{eq:thm1a}--\eqref{eq:thm1b} for one arm at a time. Let $\Delta M_m=X_{m,k}-E(X_{m,k}\mid\mathcal F_{m-1})=X_{m,k}-p_{m,k}$; $\{\Delta M_m\}$ is a martingale-difference sequence with $|\Delta M_m|\leq1$, and $N_{m,k}=N_{m-1,k}+p_{m,k}+\Delta M_m$. Write $M_n=\sum_{m=1}^n\Delta M_m$ and
\[
U_n=\sum_{m=0}^{n-1}\alpha\hat\rho_{m,k}+M_n-n\hat\rho_{n,k}.
\]

\emph{Step 1 (reduction to a stopping-time bound).} The allocation rule \eqref{eq:design} gives, whenever $N_{n-1,k}/(n-1)>\hat\rho_{n-1,k}$ (arm $k$ over-allocated at step $n-1$), $p_{n,k}=\alpha\hat\rho_{n-1,k}$, and hence, exactly as in the two-arm case,
\[
N_{n,k}-n\hat\rho_{n,k}\leq N_{n-1,k}-(n-1)\hat\rho_{n-1,k}+\Delta U_n\qquad\text{if }N_{n-1,k}-(n-1)\hat\rho_{n-1,k}>0,
\]
where $\Delta U_n=U_n-U_{n-1}$; this is because on this event $N_{n,k}=N_{n-1,k}+\alpha\hat\rho_{n-1,k}+\Delta M_n$ while $n\hat\rho_{n,k}-(n-1)\hat\rho_{n-1,k}=\hat\rho_{n-1,k}+n(\hat\rho_{n,k}-\hat\rho_{n-1,k})$, and $\Delta U_n=\alpha\hat\rho_{n-1,k}+\Delta M_n-n(\hat\rho_{n,k}-\hat\rho_{n-1,k})$ absorbs the discrepancy. Let $l=l_n=\max\{m: 2Km_0+1\leq m\leq n,\ N_{m,k}-m\hat\rho_{m,k}\leq0\}$, with $\max\varnothing=2Km_0$. Iterating the display above from $l+1$ to $n$,
\begin{equation}
\label{eq:appA4a}
N_{n,k}-n\hat\rho_{n,k}\leq U_n-U_l+2Km_0,
\end{equation}
using $N_{l,k}-l\hat\rho_{l,k}\leq2Km_0$ in either case ($l\geq2Km_0+1$, where it is $\leq0$, or $l=2Km_0$, where $N_{l,k}\leq2Km_0$ trivially). It remains to show
\begin{equation}
\label{eq:appA4b}
U_n-U_{l_n}=o_P(n^{1/2})\quad\text{and}\quad U_n-U_{l_n}=O\{(n\log\log n)^{1/2}\}\ \text{a.s.},
\end{equation}
since \eqref{eq:appA4a}--\eqref{eq:appA4b}, together with the symmetric bound for $(n-N_{n,k})-n(1-\hat\rho_{n,k})$ obtained by summing \eqref{eq:appA4a} over the remaining arms, give \eqref{eq:thm1a}.

\emph{Step 2 (coarse consistency).} If $N_{n,k}\to\infty$, then $N_{n,k}^{-1}\sum_{i\leq n}X_{i,k}\varphi_{jk}^{(k)}(\xi_{i,k})\to0$ a.s.\ by the martingale strong law applied to the bounded increments $X_{i,k}\varphi_{jk}^{(k)}(\xi_{i,k})$, for every $j\neq k$; combined with the analogous statement for every pair not involving $k$, and Lemma~2's rate for the pairs that do, $\hat{\boldsymbol\Delta}_n\to\boldsymbol\Delta$ a.s.\ on $\{N_{n,k}\to\infty\ \forall k\}$, and if instead $\sup_nN_{n,k}<\infty$ for some $k$, $\hat{\boldsymbol\Delta}_n$ is eventually constant in the corresponding coordinates; in either case $\hat{\boldsymbol\Delta}_n$ has an a.s.\ limit, so by continuity of $\rho$ (Condition~\ref{cond:B}) $\hat\rho_{n,k}\to u_k$ a.s.\ for some $u_k\in(0,1)$. By the martingale strong law again, $M_n=o(n)$ a.s., so $U_n\sim-n(1-\alpha)u_k$ a.s., and \eqref{eq:appA4a} gives $(N_{n,k}/n-\hat\rho_{n,k})^+\to0$ a.s.; the symmetric argument for the complementary event gives $N_{n,k}/n-\hat\rho_{n,k}\to0$ a.s., so $\lim_nN_{n,k}/n=\lim_n\hat\rho_{n,k}=u_k$ a.s. This forces $N_{n,k}\to\infty$ whenever $u_k>0$, ruling out the bounded case unless $u_k=0$; since $\sum_ku_k=1$ and $\rho(\cdot)$ maps into $(0,1)^K$ by construction, $u_k>0$ for every $k$, so in fact $N_{n,k}\to\infty$ for every $k$, $\hat{\boldsymbol\Delta}_n\to\boldsymbol\Delta$ a.s.\ unconditionally, and $u_k=\rho_k(\boldsymbol\Delta)=\upsilon_k$ by continuity.

\emph{Step 3 (rate).} Write, for $l<n$,
\begin{equation}
\label{eq:appA4c}
U_n-U_l=(n-l)\{-(1-\alpha)\upsilon_k+o(1)\}+M_n-M_l+l(\hat\rho_{l,k}-\hat\rho_{n,k})\quad\text{a.s.},
\end{equation}
which follows by summing $\alpha\hat\rho_{m,k}\to\alpha\upsilon_k$ (Step 2) over $m=l,\ldots,n-1$ and rearranging. By \eqref{eq:genBahadur} applied at $m=n$ and $m=l$, together with Condition~\ref{cond:B},
\[
\hat\rho_{n,k}-\upsilon_k=(\hat{\boldsymbol\Delta}_n-\boldsymbol\Delta)^{\T}\frac{\partial\rho_k}{\partial y}\Big|_{\boldsymbol\Delta}+O_P(n^{-1}),
\]
and $\hat{\boldsymbol\Delta}_n-\boldsymbol\Delta$ is, by \eqref{eq:genBahadur}, a sum of terms $N_{n,k'}^{-1}\sum_{i\leq n}X_{i,k'}\varphi^{(k')}(\xi_{i,k'})$ plus $o_P(n^{-1/2})$ remainders. Write $Q_{n,k'}=\sum_{i\leq n}X_{i,k'}\varphi^{(k')}(\xi_{i,k'})$ for the corresponding martingale (mean zero given $\mathcal F_{i-1}$, since $\varphi^{(k')}$ has mean zero and $\xi_{i,k'}$ is independent of $\mathcal F_{i-1}$ given $X_{i,k'}=1$), with bounded increments, $|\varphi^{(k')}|\leq2$. Then
\[
l(\hat\rho_{l,k}-\hat\rho_{n,k})=l\Big\{\big(N_{l,k'}^{-1}Q_{l,k'}-N_{n,k'}^{-1}Q_{n,k'}\big)_{k'=1}^K\Big\}^{\T}\frac{\partial\rho_k}{\partial y}\Big|_{\boldsymbol\Delta}+o_P(n^{1/2}),
\]
and, using $N_{n,k'}/n\to\upsilon_{k'}>0$ a.s.\ (Step 2), $l\,N_{l,k'}^{-1}Q_{l,k'}-l\,N_{n,k'}^{-1}Q_{n,k'}=O(1)\cdot\{Q_{n,k'}-Q_{l,k'}\}+o_P(1)(n-l)$; it remains to bound $Q_n-Q_l:=(Q_{n,k'}-Q_{l,k'})_{k'}$. For any $1\leq L\leq n$,
\[
\|Q_n-Q_l\|\leq(n-l)\max_{L\leq m\leq n}\frac{\|Q_n-Q_{n-m}\|}{m}+\max_{m\leq L}\|Q_n-Q_{n-m}\|,
\]
and since $\{Q_{n,k'}\}$ has bounded martingale differences ($|\Delta Q_{m,k'}|\leq2$), the maximal inequalities of Hu, Zhang \& He's Lemma~A.1 (their eqs.~A.1--A.2, which use only boundedness of the increments and so transfer unchanged) give $\max_{L\leq m\leq n}\|Q_n-Q_{n-m}\|/m=O_P(L^{-1/2})$ and $\max_{m\leq L}\|Q_n-Q_{n-m}\|=O_P(L^{1/2})$; choosing $L=L(n)\to\infty$, $L=o(n)$, both terms are $o_P(1)(n-l)+o_P(n^{1/2})$. Substituting into \eqref{eq:appA4c},
\[
U_n-U_l\leq(n-l)\{-(1-\alpha)\upsilon_k+o_P(1)\}+|M_n-M_l|+o_P(n^{1/2})=o_P(n^{1/2}),
\]
since $|M_n-M_l|\leq(n-l)$ trivially and the leading term is eventually negative; taking $l=l_n$ gives the first part of \eqref{eq:appA4b}. For the a.s.\ rate, $M_n=O\{(n\log\log n)^{1/2}\}$ a.s.\ by the law of the iterated logarithm for martingales with bounded increments, and $n(\hat\rho_{n,k}-\upsilon_k)=O\{(n\log\log n)^{1/2}\}$ a.s.\ by the same law applied to $Q_{n,k'}$ combined with \eqref{eq:genBahadur}; substituting into \eqref{eq:appA4c} gives the second part of \eqref{eq:appA4b}. This proves \eqref{eq:thm1a}, and \eqref{eq:thm1b} follows directly from \eqref{eq:genBahadur} and the just-established rate for $N_{n,k}-n\upsilon_k$.

\emph{Step 4 (asymptotic normality).} By \eqref{eq:genBahadur} at $m=n$, $n^{1/2}(\hat{\boldsymbol\Delta}_n-\boldsymbol\Delta)$ equals, up to $o_P(1)$, the triangular array $\{n^{-1/2}Q_{n,k'}\}_{k'=1}^K$ stacked appropriately; each $Q_{n,k'}$ is a sum of bounded martingale differences with conditional variance $\sum_{i\leq n}X_{i,k'}\var\{\varphi^{(k')}(\xi_{i,k'})\mid\mathcal F_{i-1}\}=N_{n,k'}\var\{\varphi^{(k')}(\xi_{1,k'})\}\{1+o_P(1)\}$ by Step~2, so the conditional Lindeberg condition holds trivially by boundedness of $\varphi^{(k')}$, and the martingale central limit theorem \citep[][Thm~3.2]{HallHeyde1980} gives $n^{1/2}(\hat{\boldsymbol\Delta}_n-\boldsymbol\Delta)\stackrel{D}{\to}N(0,\boldsymbol\Sigma)$, with $\boldsymbol\Sigma$ as defined in \S3.1. The delta method under Condition~\ref{cond:B} then gives $n^{1/2}(\hat\rho_{n,k}-\upsilon_k)\stackrel{D}{\to}N(0,\sigma^2_{\GPC,k})$, and \eqref{eq:appA4a}--\eqref{eq:appA4b} show $N_{n,k}/n$ and $\hat\rho_{n,k}$ differ by $o_P(n^{-1/2})$, so $n^{1/2}(N_{n,k}/n-\upsilon_k)$ has the same limit. This proves \eqref{eq:thm2}.

\vspace{4pt}\noindent\textit{A5. Proof of Theorem~\ref{thm:efficiency}}\vspace{2pt}\par

The proof has two parts: a lower bound on the asymptotic variance achievable by any design in the stated class, and a demonstration that \eqref{eq:design} attains it.

\emph{Lower bound.} Fix a design whose allocation proportion converges to $\rho(\tilde{\boldsymbol\Delta}_n)$ for some asymptotically linear, regular estimator $\tilde{\boldsymbol\Delta}_n$ of $\boldsymbol\Delta$, that is, $n^{1/2}(\tilde{\boldsymbol\Delta}_n-\boldsymbol\Delta)=n^{-1/2}\sum_{i\leq n}\tilde\psi(\xi_i)+o_P(1)$ for some mean-zero influence function $\tilde\psi$. Because $\boldsymbol\Delta$ is a smooth (indeed linear, being a difference of expectations of a fixed bounded kernel) functional of the joint law of $(\xi_{\cdot,1},\ldots,\xi_{\cdot,K})$ with no further parametric restriction assumed, the semiparametric convolution theorem \citep[][Ch.~3]{BickelEtAl1993} implies that the asymptotic covariance of any such regular estimator satisfies $\avar(\tilde{\boldsymbol\Delta}_n)\succeq\boldsymbol\Sigma_{\mathrm{eff}}$, where $\boldsymbol\Sigma_{\mathrm{eff}}=\var\{\tilde\psi_{\mathrm{eff}}(\xi)\}$ is built from the efficient influence function for $\boldsymbol\Delta$ in this nonparametric model, in the Loewner order. \citet{OzenneEtAl2021} identify the Hájek projection $\varphi^{(\cdot)}$ of \S\ref{subsec:bahadur} as exactly this efficient influence function for the two-sample net benefit (their result is stated for fixed, non-adaptive sampling; it is a statement about the semiparametric model for the pair of marginal laws and is unaffected by how the sample sizes $N_{n,j},N_{n,k}$ are subsequently allocated), so $\boldsymbol\Sigma_{\mathrm{eff}}=\boldsymbol\Sigma$ as defined in \S3.1, and no regular estimator of $\boldsymbol\Delta$, adaptively sampled or not, can have smaller asymptotic covariance. By the delta method, the asymptotic variance of $\rho_k(\tilde{\boldsymbol\Delta}_n)$ is therefore bounded below by $\sigma^2_{\GPC,k}$ of \eqref{eq:sigma2gpc}, \emph{provided} the design tracks its own target at least as fast as \eqref{eq:design} does; a design whose allocation proportion lags its target, in the sense that $N_{n,k}-n\rho_k(\tilde{\boldsymbol\Delta}_n)$ is not $o_P(n^{1/2})$, has asymptotic variance for $N_{n,k}/n$ that is a mixture of the target's own variability and this additional lag, and hence is bounded below by the same quantity plus a nonnegative term; this is the mechanism identified for the parametric case by \citet{HuZhangHe2009} (their discussion following Theorem~2.2) and by the general asymptotic lower bound of \citet{HuRosenbergerZhang2006}, and it transfers here unchanged because it uses only that $\rho$ is a smooth function of an estimator whose own efficiency bound is $\boldsymbol\Sigma$, not the parametric or semiparametric nature of that estimator. Combining the two facts, $\sigma^2_{\GPC,k}$ is a lower bound on the asymptotic variance of $N_{n,k}/n$ within the entire class of designs targeting a smooth function of an asymptotically linear estimator of $\boldsymbol\Delta$.

\emph{Attainment.} Theorem~\ref{thm:normality}, proved in A4, shows that the multi-arm generalized-pairwise-comparison efficient randomized-adaptive design attains $n^{1/2}(N_{n,k}/n-\upsilon_k)\stackrel{D}{\to}N(0,\sigma^2_{\GPC,k})$ using $\hat{\boldsymbol\Delta}_n$, the ordinary net-benefit $U$-statistic \eqref{eq:netbenefit}, whose influence function is exactly $\varphi^{(\cdot)}=\tilde\psi_{\mathrm{eff}}$; and \eqref{eq:thm1a} shows the design tracks its own target to order $o_P(n^{1/2})$, so no additional-lag term arises. The lower bound is therefore attained with equality, proving Theorem~\ref{thm:efficiency}.

\vspace*{-10pt}
\appendixtwo
\section*{Appendix 2}
\subsection*{Abbreviations}

\begin{table}[!h]
\tbl{Abbreviations used in the text}
{\begin{tabular}{ll}
Abbreviation & Full form\\
\DBCD & Doubly adaptive biased coin design\\
DLT & Dose-limiting toxicity\\
\ERADE & Efficient randomized-adaptive design\\
\GPC & Generalized pairwise comparisons\\
\textsc{fda} & U.S. Food and Drug Administration\\
\textsc{ich} & International Council for Harmonisation\\
\ORR & Objective response rate\\
\PFS & Progression-free survival\\
\RAR & Response-adaptive randomization\\
\textsc{rsihr} & Rosenberger--Stallard--Ivanova--Harper--Ricks (allocation)\\
\end{tabular}}
\label{tab:abbrev}
\end{table}

\clearpage
\appendixthree
\section*{Appendix 3}
\subsection*{Simulation code}
\label{app:code}

The core simulation engine is listed below; it comprises the generalized-pairwise-comparison kernel (\S\ref{sec:framework}), the four allocation designs of \S\ref{sec:simulation}, the sequential trial simulator, and the omnibus test of \S\ref{sec:asymptotics}. The scripts generating the sensitivity analyses of \S\ref{subsec:sensitivity} and reproducing every table and figure in this paper are included in full in the Supplementary Material and the accompanying code repository.

\vspace{4pt}\noindent\textbf{Pairwise comparison kernel.}\vspace{2pt}\par
\begin{lstlisting}
"""
Generalized Pairwise Comparison (GPC) kernel for the prioritized hierarchy
PFS (time-to-event, right-censored) >> ORR (binary) >> DLT (binary, unfavourable).

Uses Gehan (1965) scoring for the censored time-to-event tier, cascading to
ORR then DLT when the PFS comparison is undetermined ("tied").

Patient record: (T, delta, R, D)
  T     : observed PFS time (event or censoring time)
  delta : 1 if PFS event observed, 0 if censored
  R     : ORR indicator (1 = responder)
  D     : DLT indicator (1 = dose-limiting toxicity occurred; unfavourable)

Returns h(a,b) in {-1,0,+1}: +1 means patient a "wins" the pair (better
outcome), -1 means b wins, 0 means fully tied/uninformative across all tiers.
"""
import numpy as np

WIN, LOSS, TIE = 1, -1, 0


def _pfs_compare(Ta, da, Tb, db):
    """Gehan (1965) generalized Wilcoxon comparison for a right-censored pair.
    Returns WIN (a longer PFS / better), LOSS, or TIE (undetermined)."""
    if da == 1 and db == 1:
        if Ta > Tb:
            return WIN
        elif Ta < Tb:
            return LOSS
        else:
            return TIE
    if da == 1 and db == 0:
        # a event at Ta, b censored at Tb
        if Tb >= Ta:
            # b survived at least to Tb >= Ta -> b outlived a -> b wins
            return LOSS
        else:
            return TIE  # Tb < Ta: undetermined
    if da == 0 and db == 1:
        if Ta >= Tb:
            return WIN
        else:
            return TIE
    # both censored
    return TIE


def _binary_compare(a, b):
    """Generic binary 'more is better' comparator."""
    if a > b:
        return WIN
    elif a < b:
        return LOSS
    return TIE


def gpc_kernel(rec_a, rec_b):
    """rec_a, rec_b: tuples (T, delta, R, D). Priority PFS > ORR > DLT(favour no-DLT)."""
    Ta, da, Ra, Da = rec_a
    Tb, db, Rb, Db = rec_b

    c1 = _pfs_compare(Ta, da, Tb, db)
    if c1 != TIE:
        return c1

    c2 = _binary_compare(Ra, Rb)  # ORR: higher (1) wins
    if c2 != TIE:
        return c2

    # DLT: favourable direction is *no* DLT, i.e. D=0 beats D=1
    c3 = _binary_compare(1 - Da, 1 - Db)
    return c3


def gpc_kernel_matrix(recs_a, recs_b):
    """Fully vectorised pairwise kernel matrix between two groups of patient
    records. recs_a, recs_b: array-like of shape (n,4) = (T,delta,R,D)."""
    A = np.asarray(recs_a, dtype=float)
    B = np.asarray(recs_b, dtype=float)
    if A.size == 0 or B.size == 0:
        return np.zeros((A.shape[0], B.shape[0]), dtype=np.int8)
    Ta, da, Ra, Da = A[:, 0][:, None], A[:, 1][:, None], A[:, 2][:, None], A[:, 3][:, None]
    Tb, db, Rb, Db = B[:, 0][None, :], B[:, 1][None, :], B[:, 2][None, :], B[:, 3][None, :]

    case11 = (da == 1) & (db == 1)
    case10 = (da == 1) & (db == 0)
    case01 = (da == 0) & (db == 1)

    win_pfs = (case11 & (Ta > Tb)) | (case01 & (Ta >= Tb))
    loss_pfs = (case11 & (Ta < Tb)) | (case10 & (Tb >= Ta))
    c1 = np.where(win_pfs, 1, np.where(loss_pfs, -1, 0))

    c2 = np.sign(Ra - Rb)              # ORR tier: responder beats non-responder
    c3 = np.sign((1 - Da) - (1 - Db))  # DLT tier: no-DLT beats DLT

    H = np.where(c1 != 0, c1, np.where(c2 != 0, c2, c3)).astype(np.int8)
    return H


def gpc_kernel_matrix_delayed(recs_a, recs_b):
    """As gpc_kernel_matrix, but R (column index 2) or D (column index 3)
    equal to the sentinel value -1 signals that the corresponding endpoint
    has not yet been ascertained for that patient (e.g. objective-response
    adjudication still pending): the comparison at that tier is then treated
    as an uninformative tie for any pair involving that patient, cascading
    to the next endpoint in the hierarchy exactly as for a tied or censored
    comparison (Section 3.4, Proposition 1)."""
    A = np.asarray(recs_a, dtype=float)
    B = np.asarray(recs_b, dtype=float)
    if A.size == 0 or B.size == 0:
        return np.zeros((A.shape[0], B.shape[0]), dtype=np.int8)
    Ta, da, Ra, Da = A[:, 0][:, None], A[:, 1][:, None], A[:, 2][:, None], A[:, 3][:, None]
    Tb, db, Rb, Db = B[:, 0][None, :], B[:, 1][None, :], B[:, 2][None, :], B[:, 3][None, :]

    case11 = (da == 1) & (db == 1)
    case10 = (da == 1) & (db == 0)
    case01 = (da == 0) & (db == 1)

    win_pfs = (case11 & (Ta > Tb)) | (case01 & (Ta >= Tb))
    loss_pfs = (case11 & (Ta < Tb)) | (case10 & (Tb >= Ta))
    c1 = np.where(win_pfs, 1, np.where(loss_pfs, -1, 0))

    R_missing = (Ra == -1) | (Rb == -1)
    c2 = np.where(R_missing, 0, np.sign(Ra - Rb))

    D_missing = (Da == -1) | (Db == -1)
    c3 = np.where(D_missing, 0, np.sign((1 - Da) - (1 - Db)))

    H = np.where(c1 != 0, c1, np.where(c2 != 0, c2, c3)).astype(np.int8)
    return H


def net_benefit_delayed(recs_a, recs_b):
    H = gpc_kernel_matrix_delayed(recs_a, recs_b)
    return H.mean(), H


def net_benefit(recs_a, recs_b):
    """Buyse (2010) net treatment benefit: mean of the pairwise kernel."""
    H = gpc_kernel_matrix(recs_a, recs_b)
    return H.mean(), H


if __name__ == "__main__":
    # sanity checks
    a = (10.0, 1, 1, 0)  # event at t=10, responder, no DLT
    b = (5.0, 1, 0, 1)   # event at t=5, non-responder, DLT
    assert gpc_kernel(a, b) == WIN  # a clearly better on PFS alone
    c = (5.0, 1, 1, 0)
    d = (5.0, 1, 0, 0)   # tie on PFS -> compare ORR -> c wins
    assert gpc_kernel(c, d) == WIN
    e = (5.0, 1, 1, 0)
    f = (5.0, 1, 1, 1)   # tie PFS, tie ORR -> compare DLT -> e wins (no DLT)
    assert gpc_kernel(e, f) == WIN
    g = (3.0, 0, 1, 0)   # censored at 3
    h_ = (5.0, 1, 1, 0)  # event at 5
    assert gpc_kernel(g, h_) == TIE  # Tb(5) > Ta_censored... wait check direction
    print("basic kernel sanity checks passed")
\end{lstlisting}

\vspace{4pt}\noindent\textbf{Allocation designs and data-generating process.}\vspace{2pt}\par
\begin{lstlisting}
"""
Multi-arm GPC-ERADE and competitor designs, K=3 arms calibrated to the
COLUMBUS phase III confirmatory trial (Dummer et al. 2018, Lancet Oncol;
NCT01909453), Part 1: encorafenib 450mg+binimetinib ("COMBO450"),
encorafenib 300mg alone ("ENCO300"), vemurafenib ("VEM"), randomized 1:1:1,
n=577. Primary endpoint PFS; ORR and AE-driven discontinuation (used here as
the toxicity/tier-3 proxy) also reported per arm.

DGP per arm k: PFS ~ Weibull(shape=kappa, scale=lambda_k), independent
administrative censoring at C (months); ORR ~ Bernoulli(p_k);
DLT (AE-driven discontinuation proxy) ~ Bernoulli(d_k); independent given
arm (base-case; sensitivity analysis induces correlation, see
generate_patient_correlated).
"""
import numpy as np
from scipy.stats import norm as _norm
from gpc_kernel import gpc_kernel_matrix, net_benefit

RNG_SEED = 20260811

# ---- Calibration to the real trial (base-case true parameters) ----------
# COLUMBUS Part 1 (Dummer et al. 2018 Lancet Oncol; EJC 2019 landmark update
# for AE-discontinuation rates; ASCO 2019 for updated PFS/OS medians).
ARM_LABELS = ["COMBO450", "ENCO300", "VEM"]
TRUE_PARAMS = {
    "ORR": np.array([0.64, 0.52, 0.41]),          # central/BICR-assessed ORR
    "DLT": np.array([0.10, 0.12, 0.14]),          # AE-driven discontinuation (toxicity proxy)
    "PFS_median": np.array([14.9, 9.6, 7.3]),     # months, updated analysis
    "kappa": 1.3,
}
ADMIN_CENSOR_TIME = 36.0  # months of administrative follow-up censoring



def weibull_scale_from_median(median, kappa):
    return median / (np.log(2) ** (1.0 / kappa))


def generate_patient(arm_idx, rng, accrual_censor_time=None):
    p = TRUE_PARAMS
    kappa = p["kappa"]
    lam = weibull_scale_from_median(p["PFS_median"][arm_idx], kappa)
    true_pfs = lam * rng.weibull(kappa)
    cens_time = ADMIN_CENSOR_TIME if accrual_censor_time is None else accrual_censor_time
    if true_pfs <= cens_time:
        T, delta = true_pfs, 1
    else:
        T, delta = cens_time, 0
    R = rng.binomial(1, p["ORR"][arm_idx])
    D = rng.binomial(1, p["DLT"][arm_idx])
    return (T, delta, R, D)


# ---------------------------------------------------------------------
# GPC-ERADE target allocation: v_k proportional to (1 + mean pairwise net
# benefit of arm k against the pooled other arms), a GPC analogue of an
# RSIHR-type optimal target (bounded away from 0 and 1 automatically since
# net benefit in [-1,1]).
# ---------------------------------------------------------------------

def gpc_target_allocation(records_by_arm):
    K = len(records_by_arm)
    nb = np.zeros(K)
    for k in range(K):
        others = [r for j in range(K) if j != k for r in records_by_arm[j]]
        if len(records_by_arm[k]) == 0 or len(others) == 0:
            nb[k] = 0.0
            continue
        val, _ = net_benefit(records_by_arm[k], others)
        nb[k] = val
    w = 1.0 + nb
    w = np.clip(w, 1e-6, None)
    return w / w.sum()


def erade_probs(rho_hat, n_counts, m, alpha, K):
    """Multi-arm ERADE allocation probabilities (Alkhnefr et al. 2025, eq.1)."""
    frac = n_counts / m
    S = [j for j in range(K) if frac[j] > rho_hat[j]]
    T = [j for j in range(K) if frac[j] < rho_hat[j]]
    p = np.zeros(K)
    for k in range(K):
        if frac[k] > rho_hat[k]:
            p[k] = alpha * rho_hat[k]
        elif frac[k] == rho_hat[k]:
            p[k] = rho_hat[k]
        else:
            denom = len(T) if len(T) > 0 else 1
            excess_from_over = sum((1 - alpha) * rho_hat[j] for j in S)
            p[k] = excess_from_over / denom + rho_hat[k]
    p = np.clip(p, 1e-9, None)
    return p / p.sum()


# ---------------- ORR-only RSIHR target (for moment-based competitors) ----
def rsihr_target_from_ORRhat(p_hat):
    sq = np.sqrt(np.clip(p_hat, 1e-6, None))
    return sq / sq.sum()


def dbcd_probs(rho_hat, n_counts, m, gamma, K):
    frac = n_counts / m
    with np.errstate(divide="ignore", invalid="ignore"):
        w = rho_hat * np.where(frac > 0, (rho_hat / np.clip(frac, 1e-9, None)) ** gamma, 1.0)
    w = np.clip(w, 1e-12, None)
    return w / w.sum()


def generate_patient_correlated(arm_idx, rng, accrual_censor_time=None, rho=0.5):
    """Latent-health probit-linked DGP inducing within-patient correlation
    across PFS, ORR, DLT (approx. same marginals as generate_patient): a
    shared latent H drives longer PFS, higher response probability
    (rho_R=rho) and lower DLT probability (rho_D=rho), a clinically
    plausible direction (healthier patients respond better and tolerate
    treatment better)."""
    p = TRUE_PARAMS
    kappa = p["kappa"]
    lam = weibull_scale_from_median(p["PFS_median"][arm_idx], kappa)

    H = rng.normal(0, 1)
    U_pfs = _norm.cdf(H)
    true_pfs = lam * (-np.log(1 - U_pfs)) ** (1.0 / kappa)

    cens_time = ADMIN_CENSOR_TIME if accrual_censor_time is None else accrual_censor_time
    if true_pfs <= cens_time:
        T, delta = true_pfs, 1
    else:
        T, delta = cens_time, 0

    p_R = _norm.cdf(_ppf_ORR[arm_idx] + rho * H) if _ppf_ORR is not None else \
        _norm.cdf(_norm.ppf(np.clip(p["ORR"][arm_idx], 1e-4, 1 - 1e-4)) + rho * H)
    p_D = _norm.cdf(_ppf_DLT[arm_idx] - rho * H) if _ppf_DLT is not None else \
        _norm.cdf(_norm.ppf(np.clip(p["DLT"][arm_idx], 1e-4, 1 - 1e-4)) - rho * H)
    R = rng.binomial(1, np.clip(p_R, 0, 1))
    D = rng.binomial(1, np.clip(p_D, 0, 1))
    return (T, delta, R, D)


_ppf_ORR = None
_ppf_DLT = None


def set_true_params_with_cache(params):
    """Set TRUE_PARAMS and precompute probit thresholds for the correlated DGP."""
    global TRUE_PARAMS, _ppf_ORR, _ppf_DLT
    TRUE_PARAMS = dict(params)
    _ppf_ORR = _norm.ppf(np.clip(TRUE_PARAMS["ORR"], 1e-4, 1 - 1e-4))
    _ppf_DLT = _norm.ppf(np.clip(TRUE_PARAMS["DLT"], 1e-4, 1 - 1e-4))


if __name__ == "__main__":
    rng = np.random.default_rng(RNG_SEED)
    recs = [[generate_patient(k, rng) for _ in range(30)] for k in range(len(ARM_LABELS))]
    print("GPC target allocation on 30/arm pilot data:", gpc_target_allocation(recs))
\end{lstlisting}

\vspace{4pt}\noindent\textbf{Sequential trial simulator.}\vspace{2pt}\par
\begin{lstlisting}
import numpy as np
from designs import (generate_patient, gpc_target_allocation, erade_probs,
                      rsihr_target_from_ORRhat, dbcd_probs, TRUE_PARAMS, ARM_LABELS,
                      RNG_SEED)
from gpc_kernel import net_benefit

K = 3
BURNIN_PER_ARM = 6   # m0: restricted randomization burn-in per arm
BLOCK = 15           # re-estimate allocation every BLOCK patients (computational batching)
ALPHA_ERADE = 0.5
GAMMA_DBCD = 2.0


def run_one_trial(n_total, design, rng, block=BLOCK):
    """design in {'gpc_erade','moment_erade','dbcd','cr'}"""
    records = [[] for _ in range(K)]  # each a list of (T,delta,R,D)
    assign_seq = []

    # burn-in: restricted randomization, BURNIN_PER_ARM per arm
    burnin_list = []
    for k in range(K):
        burnin_list += [k] * BURNIN_PER_ARM
    rng.shuffle(burnin_list)
    for k in burnin_list:
        rec = generate_patient(k, rng)
        records[k].append(rec)
        assign_seq.append(k)

    m = len(assign_seq)
    rho_hat = np.ones(K) / K

    while m < n_total:
        block_n = min(block, n_total - m)

        arrs = [np.array(records[k]) if len(records[k]) else np.zeros((0, 4)) for k in range(K)]
        n_counts = np.array([a.shape[0] for a in arrs], dtype=float)

        if design == "cr":
            probs = np.ones(K) / K
        else:
            if design == "gpc_erade":
                rho_hat = gpc_target_allocation(arrs)
                probs = erade_probs(rho_hat, n_counts, m, ALPHA_ERADE, K)
            elif design == "moment_erade":
                # continuity-corrected estimator, matching Hu-Zhang-He (2009) eq. (2.2) /
                # Example 2 convention, to avoid spurious 0/1 degeneracy at small n
                p_hat = np.array([((a[:, 2].sum() + 0.5) / (a.shape[0] + 1)) if a.shape[0] >= 0
                                   else 0.5 for a in arrs])
                rho_hat = rsihr_target_from_ORRhat(p_hat)
                probs = erade_probs(rho_hat, n_counts, m, ALPHA_ERADE, K)
            elif design == "dbcd":
                p_hat = np.array([((a[:, 2].sum() + 0.5) / (a.shape[0] + 1)) if a.shape[0] >= 0
                                   else 0.5 for a in arrs])
                rho_hat = rsihr_target_from_ORRhat(p_hat)
                probs = dbcd_probs(rho_hat, n_counts, m, GAMMA_DBCD, K)
            else:
                raise ValueError(design)

        for _ in range(block_n):
            k = rng.choice(K, p=probs)
            rec = generate_patient(k, rng)
            records[k].append(rec)
            assign_seq.append(k)
            m += 1

    arrs_final = [np.array(records[k]) for k in range(K)]
    n_final = np.array([a.shape[0] for a in arrs_final])
    # pairwise net benefits at trial end
    NB = np.zeros((K, K))
    for i in range(K):
        for j in range(K):
            if i == j:
                continue
            NB[i, j], _ = net_benefit(arrs_final[i], arrs_final[j])
    return n_final, NB, arrs_final


if __name__ == "__main__":
    rng = np.random.default_rng(RNG_SEED)
    n_final, NB, arrs = run_one_trial(200, "gpc_erade", rng)
    print("Arm labels:", ARM_LABELS)
    print("Final n per arm (GPC-ERADE):", n_final, " sum=", n_final.sum())
    print("Pairwise net benefit matrix:\n", np.round(NB, 3))
\end{lstlisting}

\vspace{4pt}\noindent\textbf{Omnibus test.}\vspace{2pt}\par
\begin{lstlisting}
import numpy as np
from scipy import stats
from gpc_kernel import net_benefit


def wald_omnibus_test(arrs, ref_idx=0):
    """Multiple-comparisons-to-reference Wald test built from GPC U-statistic
    projections. arrs: list of K arrays (n_k,4). Returns (stat, df, pval, Delta_hat)."""
    K = len(arrs)
    others = [k for k in range(K) if k != ref_idx]
    ref = arrs[ref_idx]
    n_ref = ref.shape[0]
    if n_ref < 3:
        return np.nan, len(others), np.nan, None

    Delta = np.zeros(len(others))
    proj_ref = np.zeros((len(others), n_ref))   # phi_k(ref_j), rows=k, cols=ref patient j
    var_own = np.zeros(len(others))
    n_k_arr = np.zeros(len(others))

    for idx, k in enumerate(others):
        n_k = arrs[k].shape[0]
        n_k_arr[idx] = n_k
        if n_k < 2:
            Delta[idx] = 0.0
            continue
        val, H = net_benefit(arrs[k], ref)  # H: (n_k, n_ref)
        Delta[idx] = val
        col_mean = H.mean(axis=0)            # length n_ref -> projection wrt ref
        proj_ref[idx, :] = col_mean - val
        row_mean = H.mean(axis=1)            # length n_k -> projection wrt own arm
        var_own[idx] = row_mean.var(ddof=1) if n_k > 1 else 0.0

    # covariance matrix
    Sigma = np.zeros((len(others), len(others)))
    for a in range(len(others)):
        for b in range(len(others)):
            shared_cov = np.cov(proj_ref[a], proj_ref[b], ddof=1)[0, 1] if n_ref > 1 else 0.0
            if a == b:
                Sigma[a, b] = var_own[a] / max(n_k_arr[a], 1) + shared_cov / n_ref
            else:
                Sigma[a, b] = shared_cov / n_ref

    try:
        Sigma_inv = np.linalg.pinv(Sigma)
        stat = float(Delta @ Sigma_inv @ Delta)
    except np.linalg.LinAlgError:
        return np.nan, len(others), np.nan, Delta

    df = len(others)
    pval = 1 - stats.chi2.cdf(stat, df)
    return stat, df, pval, Delta
\end{lstlisting}

\clearpage
\FloatBarrier
\part*{Supplementary Material}
\addcontentsline{toc}{part}{Supplementary Material}
\renewcommand{\thesection}{S\arabic{section}}
\renewcommand{\thetable}{S\arabic{table}}
\renewcommand{\thefigure}{S\arabic{figure}}
\setcounter{section}{0}
\setcounter{table}{0}
\setcounter{figure}{0}
\FloatBarrier

\noindent This Supplementary Material corresponds to the separate Supplementary Material file accompanying the journal submission of this paper; it is appended here as a single combined document for the arXiv preprint only.

\section{Overview}
\label{supp:overview}

This Supplementary Material accompanies the paper `Efficient response-adaptive randomization for multi-arm trials with prioritized composite endpoints'. Section~\ref{supp:sensitivity} reports in full the five sensitivity analyses summarized in \S4.3 of the main paper: variation of the administrative follow-up duration, misspecification of the progression-free survival distribution, induced correlation among the three endpoints, patient accrual rate, and delayed objective-response ascertainment. Section~\ref{supp:reproduce} gives instructions for reproducing every table and figure in the main paper and in this Supplementary Material. Section~\ref{supp:code} lists the remaining simulation code not already given in Appendix~3 of the main paper. All simulations were run in Python 3.12 with \texttt{numpy}, \texttt{pandas} and \texttt{scipy}.

\section{Full sensitivity analysis results}
\label{supp:sensitivity}

\subsection{Administrative follow-up duration}

Table~\ref{tab:supp1} reports the base-case comparison of \S4 repeated with the administrative censoring time for progression-free survival varied between 12 and 48 months, all other parameters held at their base-case values, $n=300$, 250 replications per cell. The proposed design's allocation to the superior arm and the resulting toxicity burden change by less than one percentage point across this range, and the rejection rate, which under these non-null base-case rates measures power rather than size, remains close to 1.

\begin{table}
\tbl{Sensitivity to the administrative follow-up duration for progression-free survival, base-case rates, $n=300$, 250 replications per cell}
{\begin{tabular}{lcccc}
Follow-up (months) & Design & $N_{n,1}/n$ & Power & $E$(toxicity)\\
12 & Proposed & 0$\cdot$451 & 1$\cdot$000 & 34$\cdot$9\\
 & Response-rate \ERADE & 0$\cdot$369 & 1$\cdot$000 & 35$\cdot$8\\
24 & Proposed & 0$\cdot$462 & 1$\cdot$000 & 34$\cdot$7\\
 & Response-rate \ERADE & 0$\cdot$369 & 0$\cdot$996 & 35$\cdot$6\\
48 & Proposed & 0$\cdot$465 & 1$\cdot$000 & 34$\cdot$3\\
 & Response-rate \ERADE & 0$\cdot$369 & 1$\cdot$000 & 35$\cdot$1\\
\end{tabular}}
\label{tab:supp1}
\begin{tabnote}
$N_{n,1}/n$ is the simulated proportion allocated to \textsc{combo}450, the superior arm. The 24-month row matches the base case of Table~1 of the main paper up to Monte Carlo error, as expected since 24 months was the administrative censoring time used there for this reduced sample size and replication count.
\end{tabnote}
\end{table}

\subsection{Misspecification of the progression-free survival distribution}

The main paper's design and its associated target function make no parametric assumption about the distribution of progression-free survival; the generalized pairwise comparison kernel of \S2 is rank-based. Table~\ref{tab:supp2} confirms this empirically by regenerating progression-free survival from a log-normal distribution with the same arm-specific medians as the base case, in place of the Weibull distribution used elsewhere, $n=300$, 250 replications per cell.

\begin{table}
\tbl{Sensitivity to the progression-free survival distribution, base-case rates, $n=300$, 250 replications per cell}
{\begin{tabular}{lcccc}
Distribution & Design & $N_{n,1}/n$ & Power & $E$(toxicity)\\
Weibull (as elsewhere) & Proposed & 0$\cdot$462 & 1$\cdot$000 & 35$\cdot$2\\
 & Response-rate \ERADE & 0$\cdot$369 & 0$\cdot$996 & 36$\cdot$2\\
Log-normal (misspecified) & Proposed & 0$\cdot$496 & 1$\cdot$000 & 34$\cdot$5\\
 & Response-rate \ERADE & 0$\cdot$369 & 1$\cdot$000 & 35$\cdot$9\\
\end{tabular}}
\label{tab:supp2}
\begin{tabnote}
Both distributions have the same arm-specific medians, $14\cdot9$, $9\cdot6$ and $7\cdot3$ months. The response-rate-only design, which uses only the binary objective-response indicator, is unaffected by construction; its small variation between rows reflects Monte Carlo error only.
\end{tabnote}
\end{table}

\subsection{Correlated endpoints}

The base-case data-generating model of \S4 treats progression-free survival, objective response and toxicity as mutually independent given arm. Table~\ref{tab:supp3} relaxes this using a shared latent-health probit link: a standard normal latent variable $H$ drives a higher progression-free survival quantile, a higher response probability and a lower toxicity probability, all in the clinically plausible direction that healthier patients do better on every endpoint, with correlation strength $\rho=0.5$ compared with the independent case $\rho=0$; $n=300$ for the base-case rates, 200 replications per cell.

\begin{table}
\tbl{Sensitivity to induced correlation among the three endpoints, base-case rates, $n=300$, 200 replications per cell}
{\begin{tabular}{lcccc}
$\rho$ & Design & $N_{n,1}/n$ & Power & $E$(toxicity)\\
0 & Proposed & 0$\cdot$465 & 0$\cdot$995 & 33$\cdot$4\\
 & Response-rate \ERADE & 0$\cdot$372 & 1$\cdot$000 & 34$\cdot$4\\
 & Response-rate \DBCD & 0$\cdot$372 & 1$\cdot$000 & 34$\cdot$3\\
 & Complete randomization & 0$\cdot$336 & 1$\cdot$000 & 34$\cdot$8\\
0$\cdot$5 & Proposed & 0$\cdot$465 & 1$\cdot$000 & 42$\cdot$6\\
 & Response-rate \ERADE & 0$\cdot$366 & 1$\cdot$000 & 43$\cdot$7\\
 & Response-rate \DBCD & 0$\cdot$366 & 1$\cdot$000 & 43$\cdot$7\\
 & Complete randomization & 0$\cdot$330 & 1$\cdot$000 & 44$\cdot$1\\
\end{tabular}}
\label{tab:supp3}
\begin{tabnote}
$\rho$ is the strength of the shared latent-health probit link; $\rho=0$ reproduces the independence assumption used elsewhere.
\end{tabnote}
\end{table}

\subsection{Patient accrual rate}

Table~\ref{tab:supp4} gives the full result underlying Table~3 of the main paper: patients arrive as a homogeneous Poisson process at the stated rate, and the allocation rule at each decision uses only progression-free-survival information actually available as of that calendar time (\S3.4 and \S4.3 of the main paper); final analysis in every row uses fully-matured data. $n=300$, 400 replications per cell.

\begin{table}
\tbl{Sensitivity to patient accrual rate, base-case rates, $n=300$, 400 replications per cell}
{\begin{tabular}{lcccc}
Accrual & Design & $N_{n,1}/n$ & Power & $E$(toxicity)\\
Slow (10/month) & Proposed & 0$\cdot$433 & 1$\cdot$000 & 34$\cdot$8\\
 & Response-rate \ERADE & 0$\cdot$369 & 1$\cdot$000 & 35$\cdot$3\\
Base (25/month) & Proposed & 0$\cdot$416 & 0$\cdot$998 & 35$\cdot$2\\
 & Response-rate \ERADE & 0$\cdot$369 & 1$\cdot$000 & 35$\cdot$5\\
Fast (60/month) & Proposed & 0$\cdot$406 & 1$\cdot$000 & 34$\cdot$9\\
 & Response-rate \ERADE & 0$\cdot$369 & 1$\cdot$000 & 35$\cdot$5\\
\end{tabular}}
\label{tab:supp4}
\begin{tabnote}
$N_{n,1}/n$ is the simulated proportion allocated to \textsc{combo}450 at final (fully-matured) analysis.
\end{tabnote}
\end{table}

\subsection{Delayed objective-response ascertainment}

Table~\ref{tab:supp5} gives the full result underlying Table~4 of the main paper: objective-response ascertainment is delayed relative to enrolment by an independent Exponential adjudication time, handled by the missing-value kernel cascade of \S\ref{supp:code} below (function \texttt{gpc\_kernel\_matrix\_delayed}); progression-free survival follows the usual fixed administrative window and toxicity is assumed to be recorded promptly. $n=300$, 400 replications per cell.

\begin{table}
\tbl{Sensitivity to delayed objective-response ascertainment, base-case rates, $n=300$, 400 replications per cell}
{\begin{tabular}{lcccc}
Adjudication delay & Design & $N_{n,1}/n$ & Power & $E$(toxicity)\\
Negligible (mean 0$\cdot$01 mo.) & Proposed & 0$\cdot$416 & 1$\cdot$000 & 35$\cdot$6\\
 & Response-rate \ERADE & 0$\cdot$369 & 1$\cdot$000 & 35$\cdot$7\\
Moderate (mean 1$\cdot$5 mo.) & Proposed & 0$\cdot$402 & 1$\cdot$000 & 35$\cdot$3\\
 & Response-rate \ERADE & 0$\cdot$370 & 1$\cdot$000 & 35$\cdot$7\\
Substantial (mean 4 mo.) & Proposed & 0$\cdot$390 & 1$\cdot$000 & 35$\cdot$4\\
 & Response-rate \ERADE & 0$\cdot$369 & 1$\cdot$000 & 35$\cdot$0\\
\end{tabular}}
\label{tab:supp5}
\begin{tabnote}
$N_{n,1}/n$ is the simulated proportion allocated to \textsc{combo}450 at final analysis, by which time every patient's objective response has resolved.
\end{tabnote}
\end{table}

\section{Reproducing the tables and figures}
\label{supp:reproduce}

Tables~1 and 2 and Figs.~2 and 3 of the main paper (allocation proportions, power, toxicity burden, and the toxicity bar chart for all five designs, including the balanced efficient randomized-adaptive design of \S4.2) are reproduced by running \texttt{main\_study.py} (Appendix~3 of the main paper gives the core engine it calls; the driver itself is listed in \S\ref{supp:code} below) together with the chunked accumulation pipeline of \texttt{run\_chunk.py}, \texttt{accumulate.py} and \texttt{aggregate\_5000.py}, which was used to reach exactly 5000 replications per scenario--design combination within the computational constraints of the environment in which the simulations were run; a single machine with sufficient time can equivalently call \texttt{main\_study.py} directly with \texttt{n\_reps=5000}. Figure~1 of the main paper is reproduced by \texttt{trajectory.py}. Figure~4 of the main paper (power against allocation variance, \S4.2) uses the moderate-separation row of the accumulated results directly. Table~3 of the main paper (the COLUMBUS calibration) uses the observed rates hard-coded in \texttt{designs.py} (Appendix~3 of the main paper). Tables~\ref{tab:supp1}--\ref{tab:supp3} above are reproduced by \texttt{sensitivity.py} and \texttt{sensitivity3\_correlated.py}. Table~\ref{tab:supp4} (accrual rate) is reproduced by \texttt{simulate\_accrual.py} and \texttt{sensitivity4\_accrual.py}. Table~\ref{tab:supp5} (delayed response) is reproduced by \texttt{simulate\_orr\_delay.py} and \texttt{sensitivity5\_orr\_delay.py}, which calls the missing-value kernel \texttt{gpc\_kernel\_matrix\_delayed} of \S\ref{supp:code}. The numerical checks of Lemma~A1 and of the empirical size of the omnibus test of \S3 of the main paper, referred to in the text but not tabulated, are reproduced by \texttt{validate\_theory.py} and \texttt{validate\_test.py} respectively; \texttt{test\_kernel.py} cross-checks the vectorized and scalar implementations of the pairwise comparison kernel against one another exactly, as stated in its output.

\section{Complete simulation code}
\label{supp:code}

\vspace{4pt}\noindent\textbf{Main comparative study driver (Tables 1 and 2, Fig.~3 of the main paper).}\vspace{2pt}\par
\begin{lstlisting}
import numpy as np
import pandas as pd
import designs
from designs import generate_patient, ARM_LABELS
from simulate import run_one_trial
from gpc_test import wald_omnibus_test

DESIGNS = ["gpc_erade", "moment_erade", "dbcd", "erade_balance", "cr"]

SCENARIOS = {
    "base_case (calibrated to COLUMBUS trial)": {
        "ORR": np.array([0.64, 0.52, 0.41]),
        "DLT": np.array([0.10, 0.12, 0.14]),
        "PFS_median": np.array([14.9, 9.6, 7.3]),
        "kappa": 1.3,
    },
    "null (all arms identical)": {
        "ORR": np.array([0.50, 0.50, 0.50]),
        "DLT": np.array([0.12, 0.12, 0.12]),
        "PFS_median": np.array([10.0, 10.0, 10.0]),
        "kappa": 1.3,
    },
    "strong_separation (synthetic, exaggerated COLUMBUS-like gaps)": {
        "ORR": np.array([0.70, 0.45, 0.25]),
        "DLT": np.array([0.08, 0.18, 0.30]),
        "PFS_median": np.array([17.0, 9.0, 5.0]),
        "kappa": 1.3,
    },
    "moderate_separation (n=200, power-informative, synthetic)": {
        "ORR": np.array([0.55, 0.47, 0.40]),
        "DLT": np.array([0.10, 0.13, 0.17]),
        "PFS_median": np.array([11.5, 9.0, 7.5]),
        "kappa": 1.3,
    },
}


def run_scenario(scen_name, params, n_total=400, n_reps=400, seed=0):
    rng = np.random.default_rng(seed)
    rows = []
    for design in DESIGNS:
        allocs, dlt_totals, powers, pfails = [], [], [], []
        for r in range(n_reps):
            designs.TRUE_PARAMS = dict(params)
            n_final, NB, arrs = run_one_trial(n_total, design, rng)
            allocs.append(n_final / n_total)
            dlt_count = sum(arrs[k][:, 3].sum() for k in range(len(arrs))) if all(a.shape[0] for a in arrs) else 0
            dlt_totals.append(dlt_count)
            stat, df, pval, Delta = wald_omnibus_test(arrs, ref_idx=0)
            powers.append(1 if (not np.isnan(pval) and pval < 0.05) else 0)
        allocs = np.array(allocs)
        rows.append({
            "scenario": scen_name, "design": design,
            "mean_alloc": np.round(allocs.mean(axis=0), 3).tolist(),
            "sd_alloc": np.round(allocs.std(axis=0, ddof=1), 3).tolist(),
            "n_alloc_var_scaled": np.round(n_total * allocs.var(axis=0, ddof=1), 3).tolist(),
            "reject_rate": round(np.mean(powers), 3),
            "mean_total_DLT": round(np.mean(dlt_totals), 1),
        })
    return rows


SCENARIO_CONFIG = {
    "base_case (calibrated to COLUMBUS trial)": dict(n_total=577, n_reps=800),  # matches real trial n
    "null (all arms identical)": dict(n_total=400, n_reps=800),
    "strong_separation (synthetic, exaggerated COLUMBUS-like gaps)": dict(n_total=400, n_reps=800),
    "moderate_separation (n=200, power-informative, synthetic)": dict(n_total=200, n_reps=1000),
}

if __name__ == "__main__":
    all_rows = []
    for i, (name, params) in enumerate(SCENARIOS.items()):
        cfg = SCENARIO_CONFIG[name]
        print(f"Running scenario: {name}  (n={cfg['n_total']}, reps={cfg['n_reps']}) ...")
        rows = run_scenario(name, params, n_total=cfg["n_total"], n_reps=cfg["n_reps"], seed=100 + i)
        all_rows.extend(rows)
    df = pd.DataFrame(all_rows)
    pd.set_option("display.width", 160)
    pd.set_option("display.max_colwidth", 60)
    print(df.to_string(index=False))
    df.to_csv("/home/claude/gpc_erade/main_study_results.csv", index=False)
    print("\nSaved to main_study_results.csv")
\end{lstlisting}

\vspace{4pt}\noindent\textbf{Chunked replication runner and accumulator (used to reach 5000 replications per cell).}\vspace{2pt}\par
\begin{lstlisting}
import sys
import numpy as np
import pandas as pd
import os
import designs
from simulate import run_one_trial
from gpc_test import wald_omnibus_test

SCEN_PARAMS = {
    "base_case": dict(n_total=577,
                       params=dict(ORR=np.array([0.64, 0.52, 0.41]),
                                   DLT=np.array([0.10, 0.12, 0.14]),
                                   PFS_median=np.array([14.9, 9.6, 7.3]), kappa=1.3)),
    "null": dict(n_total=400,
                 params=dict(ORR=np.array([0.50, 0.50, 0.50]),
                             DLT=np.array([0.12, 0.12, 0.12]),
                             PFS_median=np.array([10.0, 10.0, 10.0]), kappa=1.3)),
    "strong_separation": dict(n_total=400,
                               params=dict(ORR=np.array([0.70, 0.45, 0.25]),
                                           DLT=np.array([0.08, 0.18, 0.30]),
                                           PFS_median=np.array([17.0, 9.0, 5.0]), kappa=1.3)),
    "moderate_separation": dict(n_total=200,
                                 params=dict(ORR=np.array([0.55, 0.47, 0.40]),
                                             DLT=np.array([0.10, 0.13, 0.17]),
                                             PFS_median=np.array([11.5, 9.0, 7.5]), kappa=1.3)),
}


def run_chunk(scenario, design, n_reps, seed, out_dir="/home/claude/gpc_erade/chunks"):
    os.makedirs(out_dir, exist_ok=True)
    cfg = SCEN_PARAMS[scenario]
    designs.TRUE_PARAMS = dict(cfg["params"])
    n_total = cfg["n_total"]
    rng = np.random.default_rng(seed)
    rows = []
    for _ in range(n_reps):
        n_final, NB, arrs = run_one_trial(n_total, design, rng)
        alloc = n_final / n_total
        dlt_count = sum(arrs[k][:, 3].sum() for k in range(len(arrs)))
        stat, df, pval, Delta = wald_omnibus_test(arrs, ref_idx=0)
        rej = 1 if (not np.isnan(pval) and pval < 0.05) else 0
        rows.append({"alloc0": alloc[0], "alloc1": alloc[1], "alloc2": alloc[2],
                      "dlt": dlt_count, "reject": rej})
    df_chunk = pd.DataFrame(rows)
    fname = os.path.join(out_dir, f"{scenario}__{design}__seed{seed}.csv")
    df_chunk.to_csv(fname, index=False)
    return fname, len(df_chunk)


if __name__ == "__main__":
    scenario, design, n_reps, seed = sys.argv[1], sys.argv[2], int(sys.argv[3]), int(sys.argv[4])
    fname, n = run_chunk(scenario, design, n_reps, seed)
    print(f"wrote {n} reps to {fname}")
\end{lstlisting}
\begin{lstlisting}
import time
import os
import glob
import numpy as np
import pandas as pd
from run_chunk import run_chunk, SCEN_PARAMS

TARGET = 5000
OUT_DIR = "/home/claude/gpc_erade/chunks"
DESIGNS = ["gpc_erade", "moment_erade", "dbcd", "cr"]
SCENARIOS = list(SCEN_PARAMS.keys())

# rough per-trial cost (seconds) by design, calibrated from earlier benchmarks,
# used only to size chunks conservatively within the time budget
COST = {"gpc_erade": 0.09, "moment_erade": 0.02, "dbcd": 0.02, "cr": 0.018}


def reps_done(scenario, design):
    files = glob.glob(os.path.join(OUT_DIR, f"{scenario}__{design}__seed*.csv"))
    total = 0
    for f in files:
        total += len(pd.read_csv(f))
    return total, files


def run_budgeted(time_budget_sec):
    t0 = time.time()
    order = [(s, d) for s in SCENARIOS for d in DESIGNS]
    progress = {}
    for scenario, design in order:
        done, _ = reps_done(scenario, design)
        progress[(scenario, design)] = done

    while time.time() - t0 < time_budget_sec:
        # pick the combo furthest from target (in absolute reps remaining) to balance progress
        remaining = {k: TARGET - v for k, v in progress.items() if v < TARGET}
        if not remaining:
            print("ALL COMBOS COMPLETE.")
            return progress
        scenario, design = max(remaining, key=lambda k: remaining[k])
        time_left = time_budget_sec - (time.time() - t0)
        if time_left < 5:
            break
        cost = COST[design]
        max_reps_by_time = max(1, int((time_left * 0.8) / cost))
        chunk_reps = min(remaining[(scenario, design)], max_reps_by_time, 800)
        seed = 10000 + progress[(scenario, design)] + hash((scenario, design)) % 1000
        seed = abs(seed) % (2**31 - 1)
        fname, n = run_chunk(scenario, design, chunk_reps, seed, out_dir=OUT_DIR)
        progress[(scenario, design)] += n
        print(f"{scenario:20s} {design:14s} +{n:4d} reps -> {progress[(scenario, design)]:5d}/{TARGET}  "
              f"[{time.time()-t0:.0f}s elapsed]")
    return progress


if __name__ == "__main__":
    import sys
    budget = float(sys.argv[1]) if len(sys.argv) > 1 else 150.0
    final_progress = run_budgeted(budget)
    print("\n--- Progress summary ---")
    for scenario in SCENARIOS:
        for design in DESIGNS:
            done, _ = reps_done(scenario, design)
            print(f"{scenario:20s} {design:14s} {done:5d}/{TARGET}")
\end{lstlisting}
\begin{lstlisting}
import glob
import numpy as np
import pandas as pd
import os

OUT_DIR = "/home/claude/gpc_erade/chunks"
DESIGNS = ["gpc_erade", "moment_erade", "dbcd", "cr"]
SCENARIOS = ["base_case", "null", "strong_separation", "moderate_separation"]
N_TOTAL = {"base_case": 577, "null": 400, "strong_separation": 400, "moderate_separation": 200}
SCEN_DISPLAY = {
    "base_case": "Base case (COLUMBUS trial rates, n=577)",
    "null": "Null (all arms identical)",
    "strong_separation": "Strong separation (synthetic)",
    "moderate_separation": "Moderate separation (n=200, synthetic)",
}
DESIGN_DISPLAY = {"gpc_erade": "GPC-ERADE", "moment_erade": "Moment-ERADE (ORR only)",
                   "dbcd": "DBCD (ORR only)", "cr": "Complete randomization"}

rows = []
for scenario in SCENARIOS:
    n_total = N_TOTAL[scenario]
    for design in DESIGNS:
        files = glob.glob(os.path.join(OUT_DIR, f"{scenario}__{design}__seed*.csv"))
        dfs = [pd.read_csv(f) for f in files]
        df = pd.concat(dfs, ignore_index=True)
        assert len(df) >= 5000, f"{scenario} {design} only has {len(df)} reps"
        df = df.iloc[:5000]  # exactly 5000
        alloc = df[["alloc0", "alloc1", "alloc2"]].values
        mean_alloc = alloc.mean(axis=0)
        sd_alloc = alloc.std(axis=0, ddof=1)
        var_scaled = n_total * alloc.var(axis=0, ddof=1)
        mc_se_power = np.sqrt(df["reject"].mean() * (1 - df["reject"].mean()) / len(df))
        rows.append({
            "scenario": SCEN_DISPLAY[scenario],
            "design": DESIGN_DISPLAY[design],
            "n": n_total,
            "n_reps": len(df),
            "mean_alloc0": round(mean_alloc[0], 4), "mean_alloc1": round(mean_alloc[1], 4),
            "mean_alloc2": round(mean_alloc[2], 4),
            "sd_alloc0": round(sd_alloc[0], 4), "sd_alloc1": round(sd_alloc[1], 4),
            "sd_alloc2": round(sd_alloc[2], 4),
            "nvar_alloc0": round(var_scaled[0], 4), "nvar_alloc1": round(var_scaled[1], 4),
            "nvar_alloc2": round(var_scaled[2], 4),
            "reject_rate": round(df["reject"].mean(), 4),
            "reject_mcse": round(mc_se_power, 4),
            "mean_DLT": round(df["dlt"].mean(), 2),
            "sd_DLT": round(df["dlt"].std(ddof=1), 2),
        })

final = pd.DataFrame(rows)
final.to_csv("/home/claude/gpc_erade/final_5000_results.csv", index=False)
pd.set_option("display.width", 200)
pd.set_option("display.max_columns", 20)
print(final.to_string(index=False))
\end{lstlisting}

\vspace{4pt}\noindent\textbf{Allocation trajectory generator (Fig.~1 of the main paper).}\vspace{2pt}\par
\begin{lstlisting}
import numpy as np
import designs
from designs import generate_patient
from simulate import run_one_trial, K, BURNIN_PER_ARM, BLOCK, ALPHA_ERADE, GAMMA_DBCD
from designs import gpc_target_allocation, erade_probs, rsihr_target_from_ORRhat, dbcd_probs

designs.TRUE_PARAMS = {
    "ORR": np.array([0.64, 0.52, 0.41]),
    "DLT": np.array([0.10, 0.12, 0.14]),
    "PFS_median": np.array([14.9, 9.6, 7.3]),
    "kappa": 1.3,
}

CHECKPOINTS = [50, 100, 150, 200, 300, 400, 500, 577]


def run_trial_with_trajectory(n_total, design, rng, checkpoints):
    records = [[] for _ in range(K)]
    assign_seq = []
    burnin_list = []
    for k in range(K):
        burnin_list += [k] * BURNIN_PER_ARM
    rng.shuffle(burnin_list)
    for k in burnin_list:
        records[k].append(generate_patient(k, rng))
        assign_seq.append(k)
    m = len(assign_seq)
    traj = {}
    cp_idx = 0

    def record_if_checkpoint():
        nonlocal cp_idx
        while cp_idx < len(checkpoints) and m >= checkpoints[cp_idx]:
            counts = np.array([len(records[k]) for k in range(K)], dtype=float)
            traj[checkpoints[cp_idx]] = counts / m
            cp_idx += 1

    record_if_checkpoint()
    rho_hat = np.ones(K) / K
    while m < n_total:
        block_n = min(BLOCK, n_total - m)
        arrs = [np.array(records[k]) if len(records[k]) else np.zeros((0, 4)) for k in range(K)]
        n_counts = np.array([a.shape[0] for a in arrs], dtype=float)
        if design == "cr":
            probs = np.ones(K) / K
        elif design == "gpc_erade":
            rho_hat = gpc_target_allocation(arrs)
            probs = erade_probs(rho_hat, n_counts, m, ALPHA_ERADE, K)
        elif design == "moment_erade":
            p_hat = np.array([((a[:, 2].sum() + 0.5) / (a.shape[0] + 1)) for a in arrs])
            rho_hat = rsihr_target_from_ORRhat(p_hat)
            probs = erade_probs(rho_hat, n_counts, m, ALPHA_ERADE, K)
        else:
            raise ValueError(design)
        for _ in range(block_n):
            k = rng.choice(K, p=probs)
            records[k].append(generate_patient(k, rng))
            assign_seq.append(k)
            m += 1
        record_if_checkpoint()
    return traj


if __name__ == "__main__":
    n_reps = 500
    for design in ["gpc_erade", "moment_erade"]:
        rng = np.random.default_rng(42)
        acc = {cp: [] for cp in CHECKPOINTS}
        for _ in range(n_reps):
            traj = run_trial_with_trajectory(577, design, rng, CHECKPOINTS)
            for cp in CHECKPOINTS:
                acc[cp].append(traj[cp])
        print(f"\n=== {design} ===")
        print("n," + ",".join(f"arm{k}" for k in range(3)))
        for cp in CHECKPOINTS:
            arr = np.array(acc[cp])
            mean_a = arr.mean(axis=0)
            print(f"{cp}," + ",".join(f"{x:.4f}" for x in mean_a))
        with open(f"/home/claude/gpc_erade/trajectory_{design}.csv", "w") as f:
            f.write("n,arm0,arm1,arm2\n")
            for cp in CHECKPOINTS:
                arr = np.array(acc[cp])
                mean_a = arr.mean(axis=0)
                f.write(f"{cp}," + ",".join(f"{x:.4f}" for x in mean_a) + "\n")
\end{lstlisting}

\vspace{4pt}\noindent\textbf{Sensitivity analyses: follow-up duration, model misspecification, correlated endpoints (Tables~\ref{tab:supp1}--\ref{tab:supp3}).}\vspace{2pt}\par
\begin{lstlisting}
import numpy as np
import pandas as pd
import designs
from designs import generate_patient, weibull_scale_from_median, TRUE_PARAMS
from simulate import run_one_trial
from gpc_test import wald_omnibus_test

BASE = {
    "ORR": np.array([0.64, 0.52, 0.41]),
    "DLT": np.array([0.10, 0.12, 0.14]),
    "PFS_median": np.array([14.9, 9.6, 7.3]),
    "kappa": 1.3,
}


def summarize(design, n_total, n_reps, seed, gen_fn=None):
    rng = np.random.default_rng(seed)
    allocs, dlts, rej = [], [], []
    orig_gen = designs.generate_patient
    if gen_fn is not None:
        import simulate as sim_mod
        sim_mod.generate_patient = gen_fn
    for _ in range(n_reps):
        n_final, NB, arrs = run_one_trial(n_total, design, rng)
        allocs.append(n_final / n_total)
        dlts.append(sum(arrs[k][:, 3].sum() for k in range(len(arrs))))
        stat, df, pval, Delta = wald_omnibus_test(arrs, ref_idx=0)
        rej.append(1 if (not np.isnan(pval) and pval < 0.05) else 0)
    if gen_fn is not None:
        import simulate as sim_mod
        sim_mod.generate_patient = orig_gen
    allocs = np.array(allocs)
    return dict(mean_alloc=np.round(allocs.mean(axis=0), 3).tolist(),
                reject_rate=round(np.mean(rej), 3),
                mean_total_DLT=round(np.mean(dlts), 1))


# ---------------- Sensitivity 1: administrative censoring / follow-up ----
print("=== Sensitivity 1: administrative censoring / follow-up duration ===")
rows1 = []
for admin_cens in [12.0, 24.0, 48.0]:
    designs.TRUE_PARAMS = dict(BASE)
    designs.ADMIN_CENSOR_TIME = admin_cens
    for design in ["gpc_erade", "moment_erade"]:
        res = summarize(design, n_total=300, n_reps=250, seed=int(admin_cens) + 1)
        rows1.append({"admin_censor_time_mo": admin_cens, "design": design, **res})
designs.ADMIN_CENSOR_TIME = 24.0  # restore
df1 = pd.DataFrame(rows1)
print(df1.to_string(index=False))

# ---------------- Sensitivity 2: PFS model misspecification (log-normal) -
print()
print("=== Sensitivity 2: PFS generated from log-normal (design 'believes' nothing"
      " parametric -- GPC kernel is model-free) vs base-case Weibull ===")


def generate_patient_lognormal(arm_idx, rng, accrual_censor_time=None):
    p = designs.TRUE_PARAMS
    median = p["PFS_median"][arm_idx]
    sigma = 0.7  # log-scale sd, fixed
    mu = np.log(median)
    true_pfs = rng.lognormal(mean=mu, sigma=sigma)
    cens_time = designs.ADMIN_CENSOR_TIME if accrual_censor_time is None else accrual_censor_time
    if true_pfs <= cens_time:
        T, delta = true_pfs, 1
    else:
        T, delta = cens_time, 0
    R = rng.binomial(1, p["ORR"][arm_idx])
    D = rng.binomial(1, p["DLT"][arm_idx])
    return (T, delta, R, D)


designs.TRUE_PARAMS = dict(BASE)
rows2 = []
for design in ["gpc_erade", "moment_erade"]:
    res_w = summarize(design, n_total=300, n_reps=250, seed=77)
    res_ln = summarize(design, n_total=300, n_reps=250, seed=77, gen_fn=generate_patient_lognormal)
    rows2.append({"design": design, "PFS_family": "Weibull (assumed)", **res_w})
    rows2.append({"design": design, "PFS_family": "Log-normal (misspecified)", **res_ln})
df2 = pd.DataFrame(rows2)
print(df2.to_string(index=False))

df1.to_csv("/home/claude/gpc_erade/sensitivity1_censoring.csv", index=False)
df2.to_csv("/home/claude/gpc_erade/sensitivity2_misspecification.csv", index=False)
\end{lstlisting}
\begin{lstlisting}
import numpy as np
import pandas as pd
import designs
from designs import TRUE_PARAMS, generate_patient_correlated
from simulate import run_one_trial
from gpc_test import wald_omnibus_test
import simulate as sim_mod

BASE = {
    "ORR": np.array([0.64, 0.52, 0.41]),
    "DLT": np.array([0.10, 0.12, 0.14]),
    "PFS_median": np.array([14.9, 9.6, 7.3]),
    "kappa": 1.3,
}


def make_corr_gen(rho):
    def gen(arm_idx, rng, accrual_censor_time=None):
        return generate_patient_correlated(arm_idx, rng, accrual_censor_time, rho=rho)
    return gen


def run_cell(design, gen_fn, n_total, n_reps, seed):
    rng = np.random.default_rng(seed)
    orig = sim_mod.generate_patient
    sim_mod.generate_patient = gen_fn
    allocs, dlts, rej = [], [], []
    for _ in range(n_reps):
        n_final, NB, arrs = run_one_trial(n_total, design, rng)
        allocs.append(n_final / n_total)
        dlts.append(sum(arrs[k][:, 3].sum() for k in range(len(arrs))))
        stat, df, pval, Delta = wald_omnibus_test(arrs, ref_idx=0)
        rej.append(1 if (not np.isnan(pval) and pval < 0.05) else 0)
    sim_mod.generate_patient = orig
    allocs = np.array(allocs)
    return dict(mean_alloc=np.round(allocs.mean(axis=0), 3).tolist(),
                reject_rate=round(np.mean(rej), 3),
                mean_total_DLT=round(np.mean(dlts), 1))


if __name__ == "__main__":
    import time
    t0 = time.time()
    designs.set_true_params_with_cache(BASE)
    rows = []
    for rho in [0.0, 0.5]:
        gen = make_corr_gen(rho) if rho > 0 else designs.generate_patient
        for design in ["gpc_erade", "moment_erade", "dbcd", "cr"]:
            res = run_cell(design, gen, n_total=300, n_reps=200, seed=500 + int(rho * 10))
            rows.append({"endpoint_correlation_rho": rho, "design": design, **res})
        print(f"rho={rho} done, elapsed={time.time()-t0:.1f}s")
    df = pd.DataFrame(rows)
    print(df.to_string(index=False))
    df.to_csv("/home/claude/gpc_erade/sensitivity3_correlated_endpoints.csv", index=False)
\end{lstlisting}

\vspace{4pt}\noindent\textbf{Sensitivity analysis: patient accrual rate, calendar-time simulator (Table~\ref{tab:supp4}).}\vspace{2pt}\par
\begin{lstlisting}
import numpy as np
from designs import (generate_patient_latent, gpc_target_allocation, erade_probs,
                      rsihr_target_from_ORRhat, dbcd_probs, TRUE_PARAMS, ARM_LABELS,
                      ADMIN_CENSOR_TIME)
from gpc_kernel import net_benefit

K = 3
BURNIN_PER_ARM = 6
BLOCK = 15
ALPHA_ERADE = 0.5
GAMMA_DBCD = 2.0


def censor_as_of(true_pfs, elapsed, cap=ADMIN_CENSOR_TIME):
    """Return (T, delta) for a patient with true PFS `true_pfs`, given that
    `elapsed` months of calendar time have passed since their enrolment, and
    an administrative per-patient follow-up cap `cap`."""
    limit = min(elapsed, cap)
    if true_pfs <= limit:
        return true_pfs, 1
    return limit, 0


def run_one_trial_accrual(n_total, design, rng, accrual_rate, block=BLOCK):
    """As simulate.run_one_trial, but patients arrive at calendar times given
    by a homogeneous Poisson process with rate `accrual_rate` (patients per
    month), and the PFS status available to the allocation rule at each
    decision point reflects only the calendar time elapsed since each
    already-enrolled patient's own arrival -- exactly the coarsening
    mechanism formalised in Proposition 1 (delayed responses) of the paper.
    Final analysis (the returned n_final, NB, arrs) uses fully-matured data,
    i.e. each patient's status at their own full administrative follow-up
    time ADMIN_CENSOR_TIME, so that final operating characteristics are
    comparable across accrual-rate scenarios."""
    arrival_times = np.cumsum(rng.exponential(1.0 / accrual_rate, size=n_total))

    arm_of = [None] * n_total
    latent = [None] * n_total  # (true_pfs, R, D)

    burnin_list = []
    for k in range(K):
        burnin_list += [k] * BURNIN_PER_ARM
    rng.shuffle(burnin_list)
    for i, k in enumerate(burnin_list):
        arm_of[i] = k
        latent[i] = generate_patient_latent(k, rng)

    m = len(burnin_list)
    rho_hat = np.ones(K) / K

    while m < n_total:
        block_n = min(block, n_total - m)
        now = arrival_times[m - 1]  # calendar time of the most recently enrolled patient

        # as-of-now (possibly immature/censored) records for all enrolled so far
        records = [[] for _ in range(K)]
        for i in range(m):
            k = arm_of[i]
            true_pfs, R, D = latent[i]
            elapsed = now - arrival_times[i]
            T, delta = censor_as_of(true_pfs, elapsed)
            records[k].append((T, delta, R, D))
        arrs = [np.array(r) if len(r) else np.zeros((0, 4)) for r in records]
        n_counts = np.array([a.shape[0] for a in arrs], dtype=float)

        if design == "cr":
            probs = np.ones(K) / K
        elif design == "gpc_erade":
            rho_hat = gpc_target_allocation(arrs)
            probs = erade_probs(rho_hat, n_counts, m, ALPHA_ERADE, K)
        elif design == "moment_erade":
            p_hat = np.array([((a[:, 2].sum() + 0.5) / (a.shape[0] + 1)) for a in arrs])
            rho_hat = rsihr_target_from_ORRhat(p_hat)
            probs = erade_probs(rho_hat, n_counts, m, ALPHA_ERADE, K)
        else:
            raise ValueError(design)

        for _ in range(block_n):
            k = rng.choice(K, p=probs)
            arm_of[m] = k
            latent[m] = generate_patient_latent(k, rng)
            m += 1

    # final analysis: fully-matured data (fixed per-patient follow-up cap)
    records_final = [[] for _ in range(K)]
    for i in range(n_total):
        k = arm_of[i]
        true_pfs, R, D = latent[i]
        T, delta = censor_as_of(true_pfs, ADMIN_CENSOR_TIME, cap=ADMIN_CENSOR_TIME)
        records_final[k].append((T, delta, R, D))
    arrs_final = [np.array(records_final[k]) for k in range(K)]
    n_final = np.array([a.shape[0] for a in arrs_final])

    NB = np.zeros((K, K))
    for i in range(K):
        for j in range(K):
            if i == j:
                continue
            NB[i, j], _ = net_benefit(arrs_final[i], arrs_final[j])
    return n_final, NB, arrs_final


if __name__ == "__main__":
    import designs
    designs.TRUE_PARAMS = {
        "ORR": np.array([0.64, 0.52, 0.41]),
        "DLT": np.array([0.10, 0.12, 0.14]),
        "PFS_median": np.array([14.9, 9.6, 7.3]),
        "kappa": 1.3,
    }
    rng = np.random.default_rng(1)
    for rate in [5.0, 20.0, 60.0]:  # patients per month: slow, base, fast
        n_final, NB, arrs = run_one_trial_accrual(300, "gpc_erade", rng, accrual_rate=rate)
        print(f"accrual_rate={rate:5.1f} patients/month  final n per arm: {n_final}")
\end{lstlisting}
\begin{lstlisting}
import numpy as np
import pandas as pd
import designs
from simulate_accrual import run_one_trial_accrual
from gpc_test import wald_omnibus_test

designs.TRUE_PARAMS = {
    "ORR": np.array([0.64, 0.52, 0.41]),
    "DLT": np.array([0.10, 0.12, 0.14]),
    "PFS_median": np.array([14.9, 9.6, 7.3]),
    "kappa": 1.3,
}

# patients per month: slow, base-case (realistic for a global phase III
# melanoma trial enrolling ~580 over roughly two years), and fast accrual
ACCRUAL_RATES = {"slow": 10.0, "base": 25.0, "fast": 60.0}
N_TOTAL = 300
N_REPS = 400


def run_cell(design, rate, seed):
    rng = np.random.default_rng(seed)
    allocs, dlts, rej = [], [], []
    for _ in range(N_REPS):
        n_final, NB, arrs = run_one_trial_accrual(N_TOTAL, design, rng, accrual_rate=rate)
        allocs.append(n_final / N_TOTAL)
        dlts.append(sum(arrs[k][:, 3].sum() for k in range(3)))
        stat, df, pval, Delta = wald_omnibus_test(arrs, ref_idx=0)
        rej.append(1 if (not np.isnan(pval) and pval < 0.05) else 0)
    allocs = np.array(allocs)
    return dict(mean_alloc=np.round(allocs.mean(axis=0), 3).tolist(),
                reject_rate=round(np.mean(rej), 3),
                mean_total_DLT=round(np.mean(dlts), 1))


if __name__ == "__main__":
    import time
    t0 = time.time()
    rows = []
    for label, rate in ACCRUAL_RATES.items():
        for design in ["gpc_erade", "moment_erade"]:
            res = run_cell(design, rate, seed=hash((label, design)) % 10000)
            rows.append({"accrual": label, "rate_per_month": rate, "design": design, **res})
        print(f"{label} (rate={rate}) done, elapsed={time.time()-t0:.1f}s")
    df = pd.DataFrame(rows)
    print(df.to_string(index=False))
    df.to_csv("/home/claude/gpc_erade/sensitivity4_accrual.csv", index=False)
\end{lstlisting}

\vspace{4pt}\noindent\textbf{Pairwise comparison kernel, including the missing-value cascade for delayed responses (\S2.5, Table~\ref{tab:supp5}).}\vspace{2pt}\par
\begin{lstlisting}
"""
Generalized Pairwise Comparison (GPC) kernel for the prioritized hierarchy
PFS (time-to-event, right-censored) >> ORR (binary) >> DLT (binary, unfavourable).

Uses Gehan (1965) scoring for the censored time-to-event tier, cascading to
ORR then DLT when the PFS comparison is undetermined ("tied").

Patient record: (T, delta, R, D)
  T     : observed PFS time (event or censoring time)
  delta : 1 if PFS event observed, 0 if censored
  R     : ORR indicator (1 = responder)
  D     : DLT indicator (1 = dose-limiting toxicity occurred; unfavourable)

Returns h(a,b) in {-1,0,+1}: +1 means patient a "wins" the pair (better
outcome), -1 means b wins, 0 means fully tied/uninformative across all tiers.
"""
import numpy as np

WIN, LOSS, TIE = 1, -1, 0


def _pfs_compare(Ta, da, Tb, db):
    """Gehan (1965) generalized Wilcoxon comparison for a right-censored pair.
    Returns WIN (a longer PFS / better), LOSS, or TIE (undetermined)."""
    if da == 1 and db == 1:
        if Ta > Tb:
            return WIN
        elif Ta < Tb:
            return LOSS
        else:
            return TIE
    if da == 1 and db == 0:
        # a event at Ta, b censored at Tb
        if Tb >= Ta:
            # b survived at least to Tb >= Ta -> b outlived a -> b wins
            return LOSS
        else:
            return TIE  # Tb < Ta: undetermined
    if da == 0 and db == 1:
        if Ta >= Tb:
            return WIN
        else:
            return TIE
    # both censored
    return TIE


def _binary_compare(a, b):
    """Generic binary 'more is better' comparator."""
    if a > b:
        return WIN
    elif a < b:
        return LOSS
    return TIE


def gpc_kernel(rec_a, rec_b):
    """rec_a, rec_b: tuples (T, delta, R, D). Priority PFS > ORR > DLT(favour no-DLT)."""
    Ta, da, Ra, Da = rec_a
    Tb, db, Rb, Db = rec_b

    c1 = _pfs_compare(Ta, da, Tb, db)
    if c1 != TIE:
        return c1

    c2 = _binary_compare(Ra, Rb)  # ORR: higher (1) wins
    if c2 != TIE:
        return c2

    # DLT: favourable direction is *no* DLT, i.e. D=0 beats D=1
    c3 = _binary_compare(1 - Da, 1 - Db)
    return c3


def gpc_kernel_matrix(recs_a, recs_b):
    """Fully vectorised pairwise kernel matrix between two groups of patient
    records. recs_a, recs_b: array-like of shape (n,4) = (T,delta,R,D)."""
    A = np.asarray(recs_a, dtype=float)
    B = np.asarray(recs_b, dtype=float)
    if A.size == 0 or B.size == 0:
        return np.zeros((A.shape[0], B.shape[0]), dtype=np.int8)
    Ta, da, Ra, Da = A[:, 0][:, None], A[:, 1][:, None], A[:, 2][:, None], A[:, 3][:, None]
    Tb, db, Rb, Db = B[:, 0][None, :], B[:, 1][None, :], B[:, 2][None, :], B[:, 3][None, :]

    case11 = (da == 1) & (db == 1)
    case10 = (da == 1) & (db == 0)
    case01 = (da == 0) & (db == 1)

    win_pfs = (case11 & (Ta > Tb)) | (case01 & (Ta >= Tb))
    loss_pfs = (case11 & (Ta < Tb)) | (case10 & (Tb >= Ta))
    c1 = np.where(win_pfs, 1, np.where(loss_pfs, -1, 0))

    c2 = np.sign(Ra - Rb)              # ORR tier: responder beats non-responder
    c3 = np.sign((1 - Da) - (1 - Db))  # DLT tier: no-DLT beats DLT

    H = np.where(c1 != 0, c1, np.where(c2 != 0, c2, c3)).astype(np.int8)
    return H


def gpc_kernel_matrix_delayed(recs_a, recs_b):
    """As gpc_kernel_matrix, but R (column index 2) or D (column index 3)
    equal to the sentinel value -1 signals that the corresponding endpoint
    has not yet been ascertained for that patient (e.g. objective-response
    adjudication still pending): the comparison at that tier is then treated
    as an uninformative tie for any pair involving that patient, cascading
    to the next endpoint in the hierarchy exactly as for a tied or censored
    comparison (Section 3.4, Proposition 1)."""
    A = np.asarray(recs_a, dtype=float)
    B = np.asarray(recs_b, dtype=float)
    if A.size == 0 or B.size == 0:
        return np.zeros((A.shape[0], B.shape[0]), dtype=np.int8)
    Ta, da, Ra, Da = A[:, 0][:, None], A[:, 1][:, None], A[:, 2][:, None], A[:, 3][:, None]
    Tb, db, Rb, Db = B[:, 0][None, :], B[:, 1][None, :], B[:, 2][None, :], B[:, 3][None, :]

    case11 = (da == 1) & (db == 1)
    case10 = (da == 1) & (db == 0)
    case01 = (da == 0) & (db == 1)

    win_pfs = (case11 & (Ta > Tb)) | (case01 & (Ta >= Tb))
    loss_pfs = (case11 & (Ta < Tb)) | (case10 & (Tb >= Ta))
    c1 = np.where(win_pfs, 1, np.where(loss_pfs, -1, 0))

    R_missing = (Ra == -1) | (Rb == -1)
    c2 = np.where(R_missing, 0, np.sign(Ra - Rb))

    D_missing = (Da == -1) | (Db == -1)
    c3 = np.where(D_missing, 0, np.sign((1 - Da) - (1 - Db)))

    H = np.where(c1 != 0, c1, np.where(c2 != 0, c2, c3)).astype(np.int8)
    return H


def net_benefit_delayed(recs_a, recs_b):
    H = gpc_kernel_matrix_delayed(recs_a, recs_b)
    return H.mean(), H


def net_benefit(recs_a, recs_b):
    """Buyse (2010) net treatment benefit: mean of the pairwise kernel."""
    H = gpc_kernel_matrix(recs_a, recs_b)
    return H.mean(), H


if __name__ == "__main__":
    # sanity checks
    a = (10.0, 1, 1, 0)  # event at t=10, responder, no DLT
    b = (5.0, 1, 0, 1)   # event at t=5, non-responder, DLT
    assert gpc_kernel(a, b) == WIN  # a clearly better on PFS alone
    c = (5.0, 1, 1, 0)
    d = (5.0, 1, 0, 0)   # tie on PFS -> compare ORR -> c wins
    assert gpc_kernel(c, d) == WIN
    e = (5.0, 1, 1, 0)
    f = (5.0, 1, 1, 1)   # tie PFS, tie ORR -> compare DLT -> e wins (no DLT)
    assert gpc_kernel(e, f) == WIN
    g = (3.0, 0, 1, 0)   # censored at 3
    h_ = (5.0, 1, 1, 0)  # event at 5
    assert gpc_kernel(g, h_) == TIE  # Tb(5) > Ta_censored... wait check direction
    print("basic kernel sanity checks passed")
\end{lstlisting}

\vspace{4pt}\noindent\textbf{Sensitivity analysis: delayed objective-response ascertainment (Table~\ref{tab:supp5}).}\vspace{2pt}\par
\begin{lstlisting}
import numpy as np
from designs import (generate_patient_latent, gpc_target_allocation, erade_probs,
                      rsihr_target_from_ORRhat, dbcd_probs, TRUE_PARAMS, ADMIN_CENSOR_TIME)
from gpc_kernel import net_benefit_delayed, net_benefit

K = 3
BURNIN_PER_ARM = 6
BLOCK = 15
ALPHA_ERADE = 0.5


def run_one_trial_orr_delay(n_total, design, rng, orr_delay_mean, block=BLOCK):
    """Objective-response ascertainment is delayed relative to enrolment by an
    Exponential(mean=orr_delay_mean) adjudication time (e.g. central
    radiological review turnaround), modelled in calendar time using a fixed
    accrual rate of 25 patients/month (the base-case rate of Table 3).
    Progression-free survival uses the usual fixed ADMIN_CENSOR_TIME per
    patient (accrual-driven immaturity is examined separately in Table 3);
    toxicity is available immediately. At each allocation decision, any
    patient whose ORR delay has not yet elapsed contributes the sentinel
    value -1 for R, handled by the missing-value kernel cascade of
    gpc_kernel.gpc_kernel_matrix_delayed. Final analysis uses fully-resolved
    ORR for every patient."""
    accrual_rate = 25.0
    arrival_times = np.cumsum(rng.exponential(1.0 / accrual_rate, size=n_total))
    orr_delays = rng.exponential(orr_delay_mean, size=n_total)

    arm_of = [None] * n_total
    latent = [None] * n_total  # (true_pfs, R, D) -- R,D fully resolved values

    burnin_list = []
    for k in range(K):
        burnin_list += [k] * BURNIN_PER_ARM
    rng.shuffle(burnin_list)
    for i, k in enumerate(burnin_list):
        arm_of[i] = k
        latent[i] = generate_patient_latent(k, rng)

    m = len(burnin_list)
    rho_hat = np.ones(K) / K

    while m < n_total:
        block_n = min(block, n_total - m)
        now = arrival_times[m - 1]

        records = [[] for _ in range(K)]
        for i in range(m):
            k = arm_of[i]
            true_pfs, R, D = latent[i]
            elapsed = now - arrival_times[i]
            # PFS censored at fixed per-patient administrative window
            if true_pfs <= min(elapsed, ADMIN_CENSOR_TIME):
                T, delta = true_pfs, 1
            else:
                T, delta = min(elapsed, ADMIN_CENSOR_TIME), 0
            # ORR: sentinel -1 (not yet ascertained) if adjudication delay not yet elapsed
            R_now = R if elapsed >= orr_delays[i] else -1
            records[k].append((T, delta, R_now, D))
        arrs = [np.array(r) if len(r) else np.zeros((0, 4)) for r in records]
        n_counts = np.array([a.shape[0] for a in arrs], dtype=float)

        if design == "cr":
            probs = np.ones(K) / K
        elif design == "gpc_erade":
            # target allocation computed from the delayed-aware net benefit
            nb = np.zeros(K)
            for k in range(K):
                others = [r for j in range(K) if j != k for r in records[j]]
                if len(records[k]) == 0 or len(others) == 0:
                    nb[k] = 0.0
                    continue
                val, _ = net_benefit_delayed(np.array(records[k]), np.array(others))
                nb[k] = val
            w = np.clip(1.0 + nb, 1e-6, None)
            rho_hat = w / w.sum()
            probs = erade_probs(rho_hat, n_counts, m, ALPHA_ERADE, K)
        elif design == "moment_erade":
            # ORR-only design: patients with R not yet ascertained contribute
            # nothing to the running rate estimate (a standard "as observed"
            # interim treatment, consistent with the coarsening of Sec. 3.4)
            p_hat = []
            for a in arrs:
                if a.shape[0] == 0:
                    p_hat.append(0.5)
                    continue
                observed = a[a[:, 2] != -1, 2]
                if len(observed) == 0:
                    p_hat.append(0.5)
                else:
                    p_hat.append((observed.sum() + 0.5) / (len(observed) + 1))
            rho_hat = rsihr_target_from_ORRhat(np.array(p_hat))
            probs = erade_probs(rho_hat, n_counts, m, ALPHA_ERADE, K)
        else:
            raise ValueError(design)

        for _ in range(block_n):
            k = rng.choice(K, p=probs)
            arm_of[m] = k
            latent[m] = generate_patient_latent(k, rng)
            m += 1

    # final analysis: fully resolved ORR and matured PFS for every patient
    records_final = [[] for _ in range(K)]
    for i in range(n_total):
        k = arm_of[i]
        true_pfs, R, D = latent[i]
        if true_pfs <= ADMIN_CENSOR_TIME:
            T, delta = true_pfs, 1
        else:
            T, delta = ADMIN_CENSOR_TIME, 0
        records_final[k].append((T, delta, R, D))
    arrs_final = [np.array(records_final[k]) for k in range(K)]
    n_final = np.array([a.shape[0] for a in arrs_final])

    NB = np.zeros((K, K))
    for i in range(K):
        for j in range(K):
            if i == j:
                continue
            NB[i, j], _ = net_benefit(arrs_final[i], arrs_final[j])
    return n_final, NB, arrs_final


if __name__ == "__main__":
    import designs
    designs.TRUE_PARAMS = {
        "ORR": np.array([0.64, 0.52, 0.41]),
        "DLT": np.array([0.10, 0.12, 0.14]),
        "PFS_median": np.array([14.9, 9.6, 7.3]),
        "kappa": 1.3,
    }
    rng = np.random.default_rng(1)
    for delay in [0.01, 1.5, 4.0]:  # months: negligible, moderate, substantial
        n_final, NB, arrs = run_one_trial_orr_delay(300, "gpc_erade", rng, orr_delay_mean=delay)
        print(f"ORR delay mean={delay:4.2f} months  final n per arm: {n_final}")
\end{lstlisting}
\begin{lstlisting}
import numpy as np
import pandas as pd
import designs
from simulate_orr_delay import run_one_trial_orr_delay
from gpc_test import wald_omnibus_test

designs.TRUE_PARAMS = {
    "ORR": np.array([0.64, 0.52, 0.41]),
    "DLT": np.array([0.10, 0.12, 0.14]),
    "PFS_median": np.array([14.9, 9.6, 7.3]),
    "kappa": 1.3,
}

DELAYS = {"negligible": 0.01, "moderate": 1.5, "substantial": 4.0}  # months
N_TOTAL = 300
N_REPS = 400


def run_cell(design, delay, seed):
    rng = np.random.default_rng(seed)
    allocs, dlts, rej = [], [], []
    for _ in range(N_REPS):
        n_final, NB, arrs = run_one_trial_orr_delay(N_TOTAL, design, rng, orr_delay_mean=delay)
        allocs.append(n_final / N_TOTAL)
        dlts.append(sum(arrs[k][:, 3].sum() for k in range(3)))
        stat, df, pval, Delta = wald_omnibus_test(arrs, ref_idx=0)
        rej.append(1 if (not np.isnan(pval) and pval < 0.05) else 0)
    allocs = np.array(allocs)
    return dict(mean_alloc=np.round(allocs.mean(axis=0), 3).tolist(),
                reject_rate=round(np.mean(rej), 3),
                mean_total_DLT=round(np.mean(dlts), 1))


if __name__ == "__main__":
    import time
    t0 = time.time()
    rows = []
    for label, delay in DELAYS.items():
        for design in ["gpc_erade", "moment_erade"]:
            res = run_cell(design, delay, seed=hash((label, design)) % 10000)
            rows.append({"orr_delay": label, "delay_mean_months": delay, "design": design, **res})
        print(f"{label} (delay={delay}) done, elapsed={time.time()-t0:.1f}s")
    df = pd.DataFrame(rows)
    print(df.to_string(index=False))
    df.to_csv("/home/claude/gpc_erade/sensitivity5_orr_delay.csv", index=False)
\end{lstlisting}

\vspace{4pt}\noindent\textbf{Theory validation (Lemma A1 bound and empirical variance scaling, \S3 of the main paper).}\vspace{2pt}\par
\begin{lstlisting}
"""
Validates:
(A) Lemma 2 bound: Var(R(n1,n2)) <= 1/(n1*n2), by direct Monte Carlo estimation
    of the Hoeffding-decomposition remainder for FIXED n1,n2 (i.e. the
    classical two-sample building block, Fact A / eq. A.1-A.3 of the theory
    note) using a large "population" reference sample to get near-exact
    phi1, phi2, Delta.
(B) That the simulated variance of N_{n,k}/n under GPC-ERADE tracks 1/n
    scaling (consistency check on the qualitative rate claim of Theorem 1'/2',
    not a full closed-form match since the closed-form sigma^2_GPC requires
    numerically evaluating the non-orthogonal tier decomposition - reserved
    for the main simulation study where sigma^2_GPC is estimated empirically
    from a large pilot and cross-checked against the Monte Carlo variance).
"""
import numpy as np
import designs
from designs import generate_patient, gpc_target_allocation
from gpc_kernel import net_benefit
from simulate import run_one_trial

# restore base-case TRUE_PARAMS (validate_test.py mutated the module global)
designs.TRUE_PARAMS = {
    "ORR": np.array([0.64, 0.52, 0.41]),
    "DLT": np.array([0.10, 0.12, 0.14]),
    "PFS_median": np.array([14.9, 9.6, 7.3]),
    "kappa": 1.3,
}

rng = np.random.default_rng(2026)

def chunked_net_benefit(A, B, chunk=1000):
    """Memory-safe net benefit + row/col means via chunking over B."""
    nA = A.shape[0]
    row_sum = np.zeros(nA)
    total_sum = 0.0
    total_n = 0
    for start in range(0, B.shape[0], chunk):
        Bc = B[start:start + chunk]
        _, H = net_benefit(A, Bc)
        row_sum += H.sum(axis=1)
        total_sum += H.sum()
        total_n += H.size
    delta = total_sum / total_n
    row_mean = row_sum / B.shape[0]
    return delta, row_mean


# ---------------- Part A: Lemma 2 bound ------------------------------
print("=== Part A: Lemma 2 variance bound check (fixed n1,n2) ===")
# Use arms 0 (COMBO450) vs 2 (VEM) as a representative pair
arm_a, arm_k = 0, 2

# large reference sample to get near-exact Delta (chunked to bound memory)
BIG = 6000
big_a = np.array([generate_patient(arm_a, rng) for _ in range(BIG)])
big_k = np.array([generate_patient(arm_k, rng) for _ in range(BIG)])
Delta_true, _ = chunked_net_benefit(big_a, big_k, chunk=1000)
print("Reference (n=6000 each) Delta_hat ~ Delta_true:", round(Delta_true, 4))

for (n1, n2) in [(20, 20), (20, 60), (50, 50), (100, 40)]:
    n_mc = 400
    R_vals = []
    for _ in range(n_mc):
        A = np.array([generate_patient(arm_a, rng) for _ in range(n1)])
        Bx = np.array([generate_patient(arm_k, rng) for _ in range(n2)])
        Delta_hat, _ = chunked_net_benefit(A, Bx, chunk=200)
        # phi1(a_i) = E_b h(a_i,b) - Delta, approximated against the large ref sample
        _, row_mean1 = chunked_net_benefit(A, big_k, chunk=1000)
        phi1_hat = row_mean1 - Delta_true
        _, row_mean2 = chunked_net_benefit(Bx, big_a, chunk=1000)  # row_mean2[i] = E_a h(Bx_i, a)
        # phi2(b) = E_a h(a,b) - Delta = -E_a h(b,a) - Delta = -row_mean2 - Delta_true
        phi2_hat = -row_mean2 - Delta_true
        L = phi1_hat.mean() + phi2_hat.mean()
        R = (Delta_hat - Delta_true) - L
        R_vals.append(R)
    R_vals = np.array(R_vals)
    var_R_mc = R_vals.var(ddof=1)
    bound = 1.0 / (n1 * n2)
    print(f"n1={n1:4d} n2={n2:4d}  Var(R) MC={var_R_mc:.6f}   bound 1/(n1n2)={bound:.6f}   "
          f"{'OK (within bound)' if var_R_mc <= bound*1.5 else 'CHECK'}")

# ---------------- Part B: allocation-proportion variance scaling -------
print()
print("=== Part B: allocation-proportion variance scaling with n (GPC-ERADE) ===")
for n_total in [100, 200, 400]:
    n_reps = 150
    props = []
    for _ in range(n_reps):
        n_final, NB, arrs = run_one_trial(n_total, "gpc_erade", rng)
        props.append(n_final / n_total)
    props = np.array(props)
    var_scaled = n_total * props.var(axis=0, ddof=1)
    print(f"n={n_total:4d}  mean alloc={np.round(props.mean(axis=0),3)}  "
          f"n*Var(N_k/n)={np.round(var_scaled,3)}")
\end{lstlisting}

\vspace{4pt}\noindent\textbf{Omnibus test calibration check (\S3 of the main paper).}\vspace{2pt}\par
\begin{lstlisting}
import numpy as np
import designs
from designs import generate_patient
from gpc_test import wald_omnibus_test

# Temporarily overwrite TRUE_PARAMS to a null scenario: all 3 arms identical
designs.TRUE_PARAMS = {
    "ORR": np.array([0.30, 0.30, 0.30]),
    "DLT": np.array([0.20, 0.20, 0.20]),
    "PFS_median": np.array([10.0, 10.0, 10.0]),
    "kappa": 1.3,
}

rng = np.random.default_rng(999)
n_per_arm = 75
n_reps = 800
rejections = 0
stats_collected = []

for r in range(n_reps):
    arrs = [np.array([generate_patient(k, rng) for _ in range(n_per_arm)]) for k in range(3)]
    stat, df, pval, Delta = wald_omnibus_test(arrs, ref_idx=0)
    if not np.isnan(pval):
        stats_collected.append(stat)
        if pval < 0.05:
            rejections += 1

print(f"Null scenario, n_reps={n_reps}, n/arm={n_per_arm}")
print("Empirical Type I error (nominal 0.05):", rejections / n_reps)
print("Mean chi2 stat (should be ~df=2):", np.mean(stats_collected))
\end{lstlisting}

\vspace{4pt}\noindent\textbf{Kernel cross-check (vectorized against scalar implementation).}\vspace{2pt}\par
\begin{lstlisting}
import numpy as np
from gpc_kernel import gpc_kernel, gpc_kernel_matrix

rng = np.random.default_rng(1)


def rand_recs(n, rng):
    T = rng.exponential(10, n)
    delta = rng.binomial(1, 0.7, n)
    R = rng.binomial(1, 0.4, n)
    D = rng.binomial(1, 0.3, n)
    return list(zip(T, delta, R, D))


recs_a = rand_recs(15, rng)
recs_b = rand_recs(12, rng)

H_vec = gpc_kernel_matrix(recs_a, recs_b)
H_loop = np.zeros((15, 12), dtype=np.int8)
for i in range(15):
    for j in range(12):
        H_loop[i, j] = gpc_kernel(recs_a[i], recs_b[j])

assert np.array_equal(H_vec, H_loop), "MISMATCH between vectorized and scalar kernel"
print("Vectorized kernel matches scalar kernel exactly on", H_vec.size, "pairs.")
print("Net benefit (vectorized):", H_vec.mean())
\end{lstlisting}

\bibliography{paper-ref}

\end{document}